\documentclass[review]{elsarticle}

\usepackage{lineno,hyperref}
\modulolinenumbers[5]

\journal{Journal of \LaTeX\ Templates}

\usepackage{graphicx}
\usepackage{url}
\usepackage{bbm}
\usepackage{caption}
\usepackage{subcaption}
\usepackage{amsfonts}
\usepackage{amssymb}
\usepackage{kantlipsum}
\usepackage{enumerate}
\usepackage{amsthm}
\usepackage[T1]{fontenc}
\usepackage{multirow}
\usepackage{amsmath}
\usepackage{bbm}
\usepackage{dsfont}
\usepackage{xcolor}
\usepackage{booktabs}
\usepackage{mathtools}
\usepackage[utf8]{inputenc}
\usepackage{enumerate}
\usepackage{algorithm}
\usepackage{algpseudocode}
\newcommand{\remove}[1]{}
\newtheorem{theorem}{Theorem}[section]
\newtheorem{lemma}{Lemma}[section]
\newtheorem{remark}{Remark}[section]

\newtheorem{prop}{Proposition}[section]

\algdef{SE}[DOWHILE]{Do}{doWhile}{\algorithmicdo}[1]{\algorithmicwhile\ #1}
\def\E{{\mathbb{E}}}
\def\R{{\mathbb{R}}}

\DeclareMathOperator*{\argmin}{arg\,min}
\algnewcommand{\algorithmicgoto}{\textbf{go to}}%
\algnewcommand{\Goto}[1]{\algorithmicgoto~\ref{#1}}

\begin{document}
%%%%%%%%%%%%%%%%%%%%%%%%%%%%%%%%%

\begin{frontmatter}

\title{Service Scheduling for Random Requests with Fixed Waiting Costs}

\author{Ramya Burra, Chandramani Singh and Joy Kuri}% <-this % stops a space
%\thanks{The authors are with the Department
%of Electronic Systems Engineering, Indian Institute of Science Bangalore, 560012 India. Emails: {\{{burra,chandra,kuri}\}@iisc.ac.in} Telephone: +91 80 22932495}% <-this % stops a space
%\thanks{The first and second authors acknowledge support from Research Fellowships of Visvesvaraya PhD Scheme and INSPIRE Faculty Research Grant (DSTO-1363).}}
\address{Department of ESE, Indian Institute of Science, Bangalore}



\cortext[mycorrespondingauthor]{Corresponding author}
%\ead{\{{burra,chandra,kuri}\}@iisc.ac.in}
\ead{burra,chandra,kuri@iisc.ac.in}
\begin{abstract}
%\color{black}
We study service scheduling problems in a slotted system in which agents arrive with service requests according to a Bernoulli process and have to leave within two slots after arrival, service costs are quadratic in service rates, and there are also waiting costs. We consider fixed waiting costs. We frame the problems as average cost Markov decision processes. While the studied system is a linear system with quadratic costs, it has state dependent control. Moreover, it also possesses a non-standard cost function structure in the case of fixed waiting costs, rendering the optimization problem complex. Here, we characterize optimal policy. We also consider a system in which the agents make scheduling decisions for their respective service requests keeping their own cost in view. We again consider fixed waiting costs and frame this scheduling problem as a stochastic game. Here, we provide Nash equilibrium.
%In the case of quadratic waiting costs, we provide an explicit expression showing that the optimal policy is linear in the system state. 
%\color{black}
\end{abstract}


\end{frontmatter}

\section{Introduction}
    Service scheduling problems have been widely studied in the literature. They apply to a wide range of applications like speed scaling in CPUs, scheduling of charging of electric vehicles (EVs), job scheduling in mobile edge computing (MEC), etc. In all these applications, service costs, measured in terms of energy consumption, increase with quantum of service. For instance, server energy consumption in cloud computing increases as a convex function of the quantum of service (see~\cite{Lin-2011-datacenters},~\cite{Ren-2012-energy-cloud}). Similarly, in the context of EV charging, the energy cost can be modelled as a quadratic function of the service offered~\cite{b15}. So, when quanta of services exceed certain thresholds, one may want to defer a part of service requests, saving energy cost in lieu of increased latency. However, large latencies must also be avoided. 

We capture the above conflicting objectives through a model having soft and hard deadlines. It is desirable to complete service requests by their soft deadlines. The service requests can be deferred beyond their soft deadlines, but then they also incur waiting costs. The waiting cost
behaves as a disincentive for deferring service to avoid excessive latencies. Of course, service requests \textit{must} be completed before their hard deadlines. We aim at deriving service scheduling policies that optimize the time average sum of service and waiting costs.
    
    Our framework is general that, as seen in Section~\ref{sec:system-model}, can be applied to many contexts like scheduling charging of EVs, job scheduling in data centers, etc.%(see Figure~\ref{fig:applications}). 
    In all these applications, both hard and soft deadlines arise naturally. For instance, an EV owner would like to get her vehicle charged at the earliest~\cite{b18} and may also have a hard deadline before which the vehicle must be charged. We discuss the applications in Section~\ref{subsec:applications}.%However,  low latency is the key performance  index for most of the real-time applications. For instance, in MEC, edge devices execute many latency-critical applications~\cite{Liu-et-al-2019}.  
   \remove{
In mobile edge computing, various latency-critical applications are pushed on to the edge of the network. These edge nodes share various resources like computing power, storage facility etc. MEC devices must offload services to neighboring MEC devices or Mobile Cloud Computing (MCC) servers to process all assigned tasks within their specified delay requirements~\cite{Liu-et-al-2019}. Hence it is very crucial to optimally schedule services on different MEC devices. Also, low latency is the key parameter index for most of the real-time applications. Therefore, it is desirable for scheduling policies to generate a low latency schedule. Similarly, from~\cite{b4} it can be noted that server power consumption in cloud computing increases as a convex function of the load. Therefore, to save on long term average power, delay-tolerant jobs need to be deferred.  In the context of EV charging, electricity cost of a charging station could rise quite steeply as the load increases. So to save on electricity cost a portion of service might be deferred. However, deferring service might result in an increase in queue length at the charging station. So, it is immensely important to minimize the service cost and penalize latency in all such cases.%Congestion control is the key to the smooth functioning of transportation networks. }
}

Optimal scheduling that intends to minimize the service costs balances service over time. However, since deferring services also incur waiting costs, balancing the quanta of services is sub-optimal. We study service scheduling in slotted systems with Bernoulli service arrivals, quadratic service costs, and service delay guarantees. We initially consider fixed waiting costs. However, in certain applications, service requests incur delay penalties depending on the amount of deferred service. To account for such a scenario, we also consider quadratic waiting costs subsequently. In particular, we consider the cases where the service requests can stay for two slots but incur fixed waiting costs in second slots. We see that this service scheduling problem is a special case of constrained linear quadratic control. We study optimal scheduling and Nash equilibria for selfish agents. These problems consider both service and waiting costs into account. We analyze optimal and equilibrium policies.
\remove{
\begin{figure}[!ht]
	
	\centering
	
	\includegraphics[width=1\textwidth]{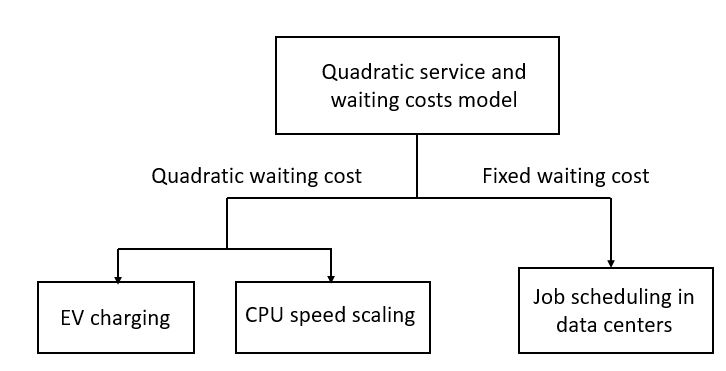}
	\caption{A few applications covered by our model.}
	\label{fig:applications}
\end{figure}
\color{black}}
\subsection{Related work}
%In case of charging EVs optimal scheduling minimizes the total charging cost.
In~\cite{Bae-Kwasinski-12}, the authors propose a centralized algorithm to minimize the total charging cost of EVs. It determines the optimal amount of charging to be received at various charging stations en route. There is another line of work which intends to minimize waiting times at the charging stations. For instance, in~\cite{Gusrialdi-et-al-14}  the authors propose a distributed scheduling algorithm that uses local information of traffic flows measured at the neighbouring charging stations to uniformly utilize charging resources along the highway and minimize the total waiting time. In our work, we consider minimizing both charging and waiting costs simultaneously. More precisely we look at two variations of waiting costs i.e., fixed waiting cost and quadratic waiting costs.
In the context of traffic routing and scheduling, the authors in~\cite{b1} consider a scenario where agents compete for a common link to ship their demands to a destination. They
obtain the optimal and equilibrium flows in the presence of polynomial congestion cost.

In~\cite{b10}, we consider routing on a ring network in the presence of quadratic congestion costs and also linear delay costs when traffic is redirected through the adjacent nodes.  However, the problems in~\cite{b10} are one-shot optimization problems as these do not have a temporal component.

Scheduling for minimizing energy costs has also been considered
in the context of CPU power consumption~\cite{b2}, big data processing~\cite{b3},
production scheduling in plants~\cite{b7}. In~\cite{b26}, the authors propose an optimal online algorithm for job arrivals with deadline uncertainty. In this work, they consider convex processing cost. They also derive competitive ratio for the proposed algorithm. None of these studies accounts for waiting costs of jobs as considered in our work.

In an earlier work~\cite{b28}, we studied service scheduling for Bernoulli job arrivals, quadratic service costs and linear waiting costs. We obtained a piece-wise linear optimal policy. We also studied Nash equilibrium in this setting. In~\cite{b27}, we extend the above study to a scenario where job sizes can take distinct values, and job arrivals constitute a Markov chain. In both these works we discuss linear waiting costs.

\color{black}
The authors in~\cite{ref2018} consider a single server slotted system with impatient customers. Impatience of customers can be seen as their having stochastic deadlines. The authors assume that the customers have geometric sojourn times but fixed one-slot service time. They consider three costs, a fixed customer holding cost per slot, a fixed cost of losing a customer, a fixed service cost, for each customer. At the beginning of each slot, if the queue is nonempty, the server has to decide whether to serve a customer. The simple service discipline and cost structure allow the authors to derive a simple rule.  The authors in~\cite{ref2012} generalized the above model by considering exponential service times and $\gamma$-Cox distributed customer sojourn times. They consider two customer classes with different arrival rates and different linear customer waiting costs but no other costs. On each service completion, if customers of both the classes are waiting, the server has to decide which customer class to choose for service.  However, the authors have only performed numerical value iteration and have obtained regions in which the first or the second customer classes are chosen. None of these works consider the case of rational customers.
\color{black}

Linear systems with quadratic cost have been widely studied in control theory. For instance,
in infinite horizon unconstrained linear quadratic control, the optimal policy is found to be
linear in system state and is given by the {\it Riccati equation}~\cite{b20}.
%The authors in~\cite{b22} study constrained linear quadratic control, and show that
%the optimal policy is piece-wise linear and the value function is piece-wise quadratic.
We have at our disposal control problems with state-dependent constraints. Moreover, in case of fixed waiting costs the problems do not conform to standard assumptions, e.g., positive definiteness of the control weighing matrix. In~\cite{eitan-shimkin-98}, the authors obtain a Nash equilibrium for a stochastic game where each arriving customer observes the current load and has to choose between a shared system whose service rate decreases with the number of customers or a constant service rate system. The optimal choice for each customer depends on the decisions of previous ones and the subsequent ones, through their effect on the current and future load in the shared server.
%We propose an optimal policy for scheduling problems where the goal is to minimize quadratic processing cost/energy along with penalizing delay. Our policy also meets the hard deadlines set by the users.
%\vspace{-0.15in}
%
\subsection{Our Contribution}
%\subsubsection*{Fixed Waiting Costs}
\begin{enumerate}
\item We study service scheduling for Bernoulli job arrivals, quadratic
service costs and service delay guarantee of two slots. For this problem, we provide the optimal scheduling policy. 
\item We then consider a scenario where each service request comes from a rational agent
who is interested in optimizing his/her own service. For this problem we obtain a symmetric Nash equilibrium of the associated stochastic game.
\end{enumerate}

\subsection{Applications and Motivation}
\label{subsec:applications}
We now illustrate how our framework can be used to model job or service scheduling problems in a variety of networks or resource sharing systems. We present an overview of applications in Table~\ref{table:applications}.
\begin{table}[htb]
%\footnotesize

%\big
\caption{Applications}
\label{table:applications}
    \centering
%\resizebox{\columnwidth}{!}{%
%\begin{tabular}{ |p{3cm}||p{4cm}|p{4cm}|  }
 \begin{tabular}{|c|c|c|}
 %\hline
%\multicolumn{3}{|c|}{List of Contribution} \\
 \hline
 Applications     & Service request & Service Cost \\
 \hline
 CPU speed scaling    & CPU cycles needed to   & Consumed energy\\
 &execute the arriving job   & \\
  \hline
  EV charging &   Energy demanded by  & Consumed energy\\
 &  the arriving vehicle  & \\
 \hline
 Job scheduling &   VM resources needed to & Consumed energy \\
in Data Centres & execute the arriving job  & \\
 %Optimal scheduling &--& Exact policy,\\
 %(General arrivals) & &(Section~VI) \\
 \hline
\end{tabular}
%}
\end{table}
\begin{enumerate}[a.]
\remove{
\item We can think of this framework as modeling traffic scheduling of transport companies where they share a network of roads~\cite{b1}. All the vehicles start at a common entry point, and are also destined to a common point. Here the agents are transport companies and their demands are amount of traffic. Further, the service cost represents congestion cost~(or, the resulting latency) on the road and there is also a waiting cost.
    %We assume that the road has large enough capacity to render the latency along it independent of the load.
    The congestion cost is linear in the traffic on the road. Note that we allow traffic to be distributed over two slots.
}
\item \textit{CPU speed scaling}: Here, the agents are jobs and  the service requests are number of CPU cycles. Further, the service cost is CPU power, which is a convex function of processor speed and there is also penalty for delaying service to a job. In \cite{b2}, the authors consider a version without delay penalties and propose off-line and on-line algorithms for minimum-energy schedule.

%\item \textit{EV charging}: Here, the agents represent vehicles. The per unit charging cost in a slot depends on the charge drawn in that slot. Each vehicle can wait up to two slots to be charged, however, the charging station also pays a cost for keeping the vehicles waiting~(say, in the form of subsidy to these vehicles). Authors in~\cite{CS-minimal-waiting} emphasizes the fact that EV users might face driver discomfort and may feel degrade in travel efficiency due to longer waiting times and longer charging period. The authors in~\cite{b18} also state that users' satisfaction depends on
% how quickly charging is completed. The charging station would like to minimize the combination of charging and waiting costs.  In a scenario where each vehicle owner would like to minimize the sum of its charging and waiting costs, we arrive at a non-cooperative game among the owners.
\item \textit{Job scheduling in Data Centres}: Here, the agents are jobs and the service requests are Virtual Machine(VM) resources that are specified in terms of CPU power, storage etc.~\cite{b3},~\cite{Data-center-routing},~\cite{Edge-Computing-System}. Jobs need to be served by fixed deadlines to meet their service level agreements~(SLAs). Therefore, the service cost is CPU power. Moreover, the jobs need fixed storage for their entire execution times in the system. Our fixed waiting costs can be used to account for the additional storage costs beyond the slots in which the jobs arrive. Note that our formulation assumes that the jobs can wait for at most one extra slot and they all have same CPU and storage requirements.
\end{enumerate}
\subsubsection{Motivation for Different Performance Criteria} 
In many cases, network (or, resource) managers schedule service requests to optimize time-average service and waiting costs while respecting their deadlines. For instance, in the examples of job scheduling in CPUs or in data centers, service schedulers may want to optimize average power and storage costs. These objectives are captured by the proposed optimal scheduling problem.   

On the other hand, in some contexts the strategic agents who bring service requests to the system dictate their service schedules. Their scheduling decisions are aimed at minimizing their respective service and waiting costs. Such scenarios can naturally be modeled using non-cooperative stochastic games. \remove{For instance, in the transportation network example if each transport company is interested in minimizing scheduling cost of only its own vehicles, a stochastic game emerges~\cite{b1}. }For instance, if the EV owners in the EV charging example strive to minimizing their respective charging and waiting costs a stochastic game emerges.

%\subsubsection{Motivation for Different Waiting Cost Structures}
\subsubsection{Motivation for Fixed Waiting Cost Structure}In several systems of interest, agents can enter the system or leave only at slot boundaries, e.g., from~\cite{LangTong_scheduling_commitment}, compute tasks derive utility only at slot boundaries. In such tasks that complete only at slot boundaries, the current operating job will be present in the system until its next slot boundary irrespective of the amount of pending service. Thus the waiting cost is fixed and does not depend on the amount of deferred service. Similarly, in data centers, the job in execution would hold a certain amount of fixed storage~\cite{Data-center-routing}. That storage is not released till the job exits the system. Thus we intend to capture the fixed storage costs in fixed waiting costs. In some other systems, service requests have soft deadlines; missing soft deadlines is tolerable but not desirable. The
authors in ~\cite{Soft_deadline} propose the notion of \textit{tardiness} which is the difference between the service requests’ actual service completion times and their soft deadlines. In our formulation, each request has a soft deadline of one slot and a hard deadline of two slots. The fixed waiting cost models the tardiness of a service request that is not completely served in its first slot. These scenarios motivate fixed waiting costs proposed in Section~\ref{sec:system-model-fixed-cost}. 

\remove{
In several systems, agents are admitted at slot boundaries, but they can
leave soon as their services are complete, e.g., consider EVs at EV Charging stations. Then the waiting period of an agent can depend on the amount of the deferred service. It is reasonable to consider waiting costs that depend on the amount of the deferred service in such cases. In~\cite{EV_waiting_cost}, the authors introduce a non-decreasing convex penalty on EVs' average waiting time. We have considered linear waiting costs in our earlier work~\cite{b28}. The authors in~\cite{quadratic_EV_dissatisfaction} consider systems where the service facility need not deliver complete service leading to dissatisfaction of agents which they capture using a dissatisfaction cost - a quadratic function of the unfinished amount of service. In a similar context, the authors in~\cite{LangTong_Restless_MAB} consider a convex dissatisfaction cost. Quadratic waiting costs capture users' higher sensitivity to incremental delays while still rendering the problems in the class of linear systems with quadratic costs. In Sections~\ref{sec:quad-waiting-cost} and ~\ref{sec:nash-equilibrium} we alter the system model of Section~\ref{sec:system-model-fixed-cost} to have quadratic waiting costs instead of fixed waiting costs and analyze the resulting scheduling problems.}
\begin{table}[htb]
%\footnotesize

%\big
\caption{Characteristics of Fixed waiting costs}
\label{table:Various-waiting-costs}
    \centering
%\resizebox{\columnwidth}{!}{%
%\begin{tabular}{ |p{3cm}||p{4cm}|p{4cm}|  }
 \begin{tabular}{|c|c|}
 %\hline
%\multicolumn{3}{|c|}{List of Contribution} \\
 \hline
    & Fixed waiting cost\\
 \hline
Best suited    & Departures happening only at   \\
 &slot boundaries   \\
  \hline
  Sensitivity towards &   Not sensitive (fixed cost for \\
 service deferred &  any positive deferred service)   \\
 \hline
Application &   Job scheduling  \\
 & in Data Centres  \\
 %Optimal scheduling &--& Exact policy,\\
 %(General arrivals) & &(Section~VI) \\
 \hline
\end{tabular}
%}
\end{table}
%%\vspace{-0.1in}
\color{black} 

\remove{
\subsubsection*{Quadratic Waiting Costs}
\begin{enumerate}
\item We also study optimal scheduling in the presence of quadratic waiting costs. Here we derive the optimal scheduling policy. %We also provide an algorithm that yields the optimal control for general service requirements.
\item We obtain a symmetric Nash equilibrium for the associated stochastic game.
\end{enumerate}
}
We also present a comparative numerical study to illustrate the impact of various waiting cost structures and performance criteria~(optimal scheduling vs strategic scheduling by selfish agents). 

Table~\ref{table:contributions} shows organization of our contribution.

\begin{table}[htb]
%\footnotesize
%\color{black}
%\big
\caption{Organization of contributions}
\label{table:contributions}
  \centering
%\resizebox{\columnwidth}{!}{%
%\begin{tabular}{ |p{3cm}||p{4cm}|p{4cm}|  }
 \begin{tabular}{|c||c|c|}
 %\hline
%\multicolumn{3}{|c|}{List of Contribution} \\
 \hline
 Versions     & Fixed Waiting Cost \\
 \hline
 Optimal scheduling    & Special cases: Section ~\ref{subsec:special_cases}  \\
 & General case: Section~\ref{sec:fixed-waiting-cost-total-optimal}    \\
 
  \hline
 Nash equilibrium &   Special cases: Section~\ref{subsec:nash-special-case}

   \\
  &  General case: Section~\ref{sec:fixed-waiting-cost-total-nash} \\
 \hline
 %Optimal scheduling &--& Exact policy,\\
 %(General arrivals) & &(Section~VI) \\
 %\hline
\end{tabular}
%}
\end{table}
\section{System Model}
\label{sec:system-model}
We consider a time-slotted system where time is divided into discrete slots. \color{black}The length of the slot depends on the application, e.g., in the case of CPU speed scaling the slots are of the order of $ms$ where in the case of job scheduling the slots many of the order of several tens of minutes. \color{black} Agents arrive over slots to a service facility. Every agent is characterized by its arrival time, deadline, and the amount of service it requires. Each service request has to be wholly served before its deadline. So service can be scheduled such that portions of the agents’ required service are served in the future slots before their respective deadlines. Serving requests incur a cost, with the cost per unit service in a slot depending on the quantum of service delivered in that slot. \color{black}Though the service facility has enough capacity to serve all the agents in the system, some of the service may be deferred to save on the service cost. \color{black}We consider two scheduling problems: one where the service facility makes scheduling decisions to optimize the overall time-average cost and the other where the agents make scheduling decisions for their respective service requests to minimize their costs. Below we present the system model and both the problems formally.
\subsection{Service request model} Agents with service requests arrive according to an i.i.d. Bernoulli$(p)$ process; $~ p \in (0,1)$. We assume that all the agents require equal amount of service, denoted as $\psi$. Further, each request can be met in at most two slots, i.e., a fraction
the service request arriving in a slot could be deferred to the next slot. \color{black}As every agent leaves at the end of two slots, in any slot there can be a maximum of two agents. Hence the system remains stable. It is assumed that the service facility can serve up to $2\psi$ units in a slot. \color{black} %The words agents and service requests are used interchangeably in the rest of the paper.
%\color{black}
%\vspace{-0.2in}
\subsection{Cost model} \label{sec:system-model-fixed-cost}The cost consists of two components:
\begin{itemize}
\item {\it Service cost:} \color{black}The service cost per unit service in a slot is a linear function of the total service
offered in that slot. Thus the total service cost in a slot is square
of the total offered service in that slot. \color{black}For instance, in the context of EV charging, per unit electricity cost is modelled as a linear function of the load~\cite{b15},~\cite{He_Linear_perunit_electricity}.
\item {\it Waiting cost:} Each service incurs a fixed waiting cost $d >0$ when a portion of the service is deferred to the next slot. This waiting cost can be interpreted as the penalty for not serving the service request in the same slot in which it has arrived. We introduce the waiting cost to strike a balance between service cost and latency. The constant $d$ can be seen as relative weight of waiting cost vis-s-vis service cost for instance e.g., higher $d$ indicates that the users are more sensitive to latency. %We also study quadratic waiting costs which are described in appropriate sections later.
\end{itemize}
Let, for $k \geq 1$,  $x_k$ be the remaining demand from slot $k-1$ to
slot $k$; $x_1 = 0$. This demand must be met in slot $k$.
Also, for $k \geq 1$, let $v_k$ be the extra service offered in slot $k$ over $x_k$.
Clearly, $v_k \in [0,\psi]$ and is $0$ if there is no new request in slot $k$.
A {\em scheduling policy} $\overline{\pi} = (\pi_k,k \geq 1)$
is a sequence of functions $\pi_k:[0,\psi] \rightarrow [0,\psi]$
such that if there is a service request in slot $k$ then $\pi_k(x_k)$ gives the amount of service deferred
from slot $k$ to slot $k+1$.
In other words,
\begin{equation*}
x_{k+1} = \begin{cases}
\pi_k(x_k) = \psi - v_k, & \text{if a request arrives in slot $k$}, \\
0, & \text{otherwise.}
\end{cases}
\end{equation*}
We consider the following two scheduling problems.

\color{black}
\subsubsection{Optimal Scheduling}
\label{subsec:opt}
We aim to minimize the time-averaged cost of the service facility. Here, waiting cost is imposed by the service facility to reduce the latency of the individual service requests. More precisely, we want to determine the scheduling policy $\overline{\pi}$
that minimizes
\begin{equation}
\lim_{T \to \infty} \frac{1}{T}\sum_{k=1}^T\E[(x_k + v_k)^2 + d \mathbbm{1}_{\{v_{k}\in (0,\psi) \}}].
\label{eqn:average-cost}
\end{equation}
We obtain the optimal solution in Section ~\ref{sec:fixed-waiting-cost}.

%\begin{remark}
At first glance,  the optimization problem appears
to be a special case of the well-studied {\em constrained linear
quadratic control Markov decision problems}. In particular, if we define binary variables
$e_k, k\geq 1$, as
\begin{equation*}
e_k = \begin{cases}
\psi, & \text{if slot $k$ has a request}, \\
0, & \text{otherwise},
\end{cases}
\end{equation*}
then $(x_k,e_k)$ can be considered to be the system state in slot
$k$. The total service in slot $k$, $\bar{u}_k \in [x_k,x_k+e_k]$,
and $w_k = e_{k+1}$ can be considered the action and the noise
in slot $k$, respectively. Then state evolution happens as
$(x_{k+1},e_{k+1}) = (x_k + e_k - \bar{u}_k, w_k)$
and the single stage cost is $d \mathbbm{1}_{x_k + e_k - \bar{u}_k>0} + \bar{u}_k^2$. {\em We see that
the actions are subject to state dependent constraints
and the single stage costs are not expressible in the form
$(x_k,e_k)^TQ(x_k,e_k) + \bar{u}_k^2$ with $Q$ a
positive semidefinite matrix}. Thus the problem
does not conform to the standard framework.
%\end{remark}

%\color{black}
\remove{In the context of CPU speed scaling, the parameters introduced above could be mapped as follows.
\begin{table}[htb]
%\footnotesize
\color{black}
%\big
\caption{CPU speed scaling}
\label{table:contributions}
  \centering
%\resizebox{\columnwidth}{!}{%
%\begin{tabular}{ |p{3cm}||p{4cm}|p{4cm}|  }
 \begin{tabular}{|c|c|}
 %\hline
%\multicolumn{3}{|c|}{List of Contribution} \\
 \hline
 $x_k$   & Number of CPU cycles pending in slot $k$ from the job arrived in slot $k-1$   \\
 \hline
 $v_k$   & Number of CPU cycles offered in slot $k$ to the job arrived in slot $k$   \\
  \hline
 $e_k$   & Number of CPU cycles requested in slot $k$ by the job that arrived in slot $k$   \\
 \hline
 %Optimal scheduling &--& Exact policy,\\
 %(General arrivals) & &(Section~VI) \\
 %\hline
\end{tabular}
%}
\end{table}\color{black}}

\subsubsection{Equilibrium for Selfish Agents}
\color{black}Recall that, in our model each agent comes with a service request, all service requests being of the same size. \color{black}Here, we consider rational agents, each determining how much of its request should be deferred. Further, each agent is aiming at minimizing his/her own service and waiting costs.
We can frame this problem as a non-cooperative dynamic game among the agents. Here, the waiting cost is imposed by every individual agent in the system to minimize their respective waiting times. In this context, let us refer to $\pi_k$ as a strategy
of the agent who arrives in slot $k$~(if there is one) and $\overline{\pi} = (\pi_k,k \geq 1)$
as a strategy profile.\footnote{Notice that $\pi$ consists of a strategy
for each slot but there may not be any agent in a slot to use the corresponding strategy.} \color{black}If an agent $k$ sees the system state as $x$, then the agent chooses the action $\pi_k(x)$. Then the total demand served in that slot is~$x+\psi-\pi_k(x)$, which is per unit cost. Therefore, the total service cost levied on \textit{the agent} is $(\psi-\pi_k(x))(x+\psi-\pi_k(x))$.\color{black}
The expected cost of an agent who arrives in slot $k$, if it sees a remaining demand $x$,
is
\begin{align}
c_k(x,\overline{\pi}) = ~&(\psi-\pi_k(x))(\psi-\pi_k(x)+x) + \pi_k(x)(\pi_k(x)+ p (\psi-\pi_{k+1}(\pi_k(x))) )\nonumber \\
&+d \mathbbm{1}_{\pi_k(x)>0} .\label{eqn:selfish-agent-cost-expression}
\end{align}
A strategy profile $\overline{\pi}$ is called a {\it Nash
equilibrium} if
\[
c_k(x,\overline{\pi}) \leq c_k(x,(\mu,\overline{\pi}_{-k}))
\]
for all $k \geq 1$, $x \in [0,\psi]$ and strategies $\mu:[0,\psi] \to [0,\psi]$.
\footnote{$(\mu,\overline{\pi}_{-k}) \triangleq (\pi_1,\dots,\pi_{k-1},\mu,\pi_{k+1},\dots)$.}
We focus on symmetric Nash equilibria of the form $(\pi,\pi,\dots)$
and obtain one such equilibrium in Section~\ref{sec:nash-equilibrium-fixed}.

\color{black}In the context of Job scheduling in data centers, the parameters introduced above could be mapped as follows.
\begin{enumerate}
    \item $x_k$: CPU power pending in slot $k$ for the job arrived in slot $k-1$.
    \item $v_k$: CPU power offered in slot $k$ to the job arrived in slot $k$.
    \item $e_k$: Total CPU power requested in slot $k$ by the job that arrived in slot $k$.
\end{enumerate}
\color{black}

\section{OPTIMAL SCHEDULING}
\label{sec:fixed-waiting-cost}
We first show that the optimal scheduling problem can
be transformed into a stochastic shortest path problem.
Let $A_i, i \geq 1$ be the successive slots that
have service requests but do not have service
requests in the preceding slots. More precisely,
\begin{equation*}
A_i = \begin{cases}
\min \{k: \text{slot $k$ has a request}\}, & \text{ if } i = 1, \\
 \min \left\{k > A_{i-1}: \text{slot $k$ has a request but}\right.\\
\left. \hspace{0.3in}\text{$k-1$ does not}\right\}, & \text{ if } i \geq 2.\\
\end{cases}
\end{equation*}
Then $A_i, i \geq 1$ can be seen to be {\it renewal instants}
of a delayed renewal process. \color{black}The following lemma gives the mean of renewal lifetimes, $A_{i+1} - A_i, i \geq 1$. \color{black}
\begin{lemma}
\label{lemma:expectation-Ai}
 $\E(A_{i+1}-A_i)=\frac{1}{p(1-p)}.$
\end{lemma}
\begin{proof}
%Proof is exactly same as the proof in~\cite[Section B in supplementary material]{our-journal}.
See~\ref{Appendix:expectation-Ai}.
\end{proof}
Hence, from the {\it Renewal Reward Theorem}~\cite{renewal-reward-theorem},
\begin{align*}
\lefteqn{\lim_{T \to \infty} \frac{1}{T}\sum_{k=1}^T\E[(x_k + v_k)^2 + d \mathbbm{1}_{\{v_{k}\in (0,\psi) \}}]}\\
  &= \frac{\E\left[\sum_{k = A_i}^{A_{i+1}-1}\left((x_k + v_k)^2 + d \mathbbm{1}_{\{v_{k}\in (0,\psi) \}}\right)\right]}{\E[A_{i+1} - A_i]} \\
 &= p(1-p)\E\left[\sum_{k = A_i}^{A_{i+1}-1}\left((x_k + v_k)^2 + d \mathbbm{1}_{\{v_{k}\in (0,\psi)\}}\right)\right].
\end{align*}
 So, we can focus on minimizing the aggregate
 cost over a ``renewal lifetime'' $A_{i+1} - A_i$.
But we do not incur any cost after service completion
of the last customer in this lifetime. We can thus
frame the problem as {\it stochastic shortest path problem} where terminal state corresponds to
absence of request in a slot.
\remove{
\begin{remark}
Average cost optimality problems have equivalent stochastic shortest formulations and we often solve the latter ones to get a solution to the former ones~\cite{b21}. We obtain a simpler connection in the service scheduling problem as renewal cycle length does not depend on policy.
\end{remark}
}
\paragraph*{Stochastic shortest path formulation} We let $x_k$ be the system state
at any slot $k$ and $t$ be a special {\it terminal state}
which is hit if there is no new request in a slot.
Let $x_{k+1}$ also denote the action in slot $k$.
Clearly, the single stage cost before hitting
the terminal state is $(x_k + \psi-x_{k+1})^2 + d\mathbbm{1}_{\{x_{k+1} > 0\}} $.
Given the state-action pair in slot $k$, $(x_k,x_{k+1})$, the next state
is the terminal state with probability $1-p$ and the terminal
cost is $x_{k+1}^2$.

Let $J:[0, \psi] \rightarrow \R_+$ be the optimal cost function
for the problem. It is the solution of the following Bellman's equation: For all $x \in [0,\psi]$,
\begin{align*}
J(x) =  \min\left\{(\psi+x)^2 + p J(0), \right.&\min_{u \in (0, \psi]} \{(\psi-u+x)^2 
+ d + pJ(u) +(1-p)u^2\}\}.
\end{align*}
Notice that the term under the inner minimization at $u = 0$ exceeds the first term by $d$. Hence we can change the constraint on $u$ in the inner minimization to $[0,\psi]$ without altering the solution $J(\cdot)$. In other words, $J(\cdot)$ is also the solution to the following equation:
\begin{align}
J(x) =  \min\left\{(\psi+x)^2 + p J(0),\right. &\min_{u \in [0, \psi]} \{(\psi-u+x)^2 
+ d + pJ(u) &\left. \left. +(1-p)u^2\right\}\right\}. \label{eqn:fixed-cost-bellman}
\end{align}
The optimal cost is attained by a stationary policy of the
form $(\pi^{\ast},\pi^{\ast},\dots)$ where $\pi^{\ast}(x)$ minimizes
the right hand side in the above equation for all $x$. For brevity, we use $\pi^{\ast}$
to refer to this policy. Let us define the "$k$-stage problem" as the one that allows at most $k+1$
 service requests. More precisely, here the system is {\it forced to enter}
 the terminal state after $k+1$ service requests if it has not already done so.
 Let $J_k(\cdot)$ be the optimal cost function
 of the $k$-stage problem. Clearly,
\remove{
Let $\pi^{\ast}$ be the optimal stationary policy for this problem.
Let us also define the "$k$-stage problem" as in Section~\ref{sec:Optimal-scheduling} and call the corresponding optimal cost function $J_k(\cdot)$. Then}
\begin{equation}
J_0(x)=\min\{(\psi+x)^2,\min_{u\in [0,\psi]}\{(\psi+x-u)^2+d+u^2\}\}
\label{eqn:j-prime-0}
\end{equation}
and
\begin{align}
J_k(x) = \min\{(\psi+x)^2 &+pJ_{k-1}(0), \min_{u\in [0,\psi]}\{(\psi+x-u)^2  d+pJ_{k-1}(u)+(1-p)u^2\}\}. \label{eqn:j-prime-k}
\end{align}
\remove{
Let $\pi_k(\cdot)$ be the optimal controls
of the $k$-stage problems~(i.e., optimal controls in~\eqref{eqn:j-prime-0}-\eqref{eqn:j-prime-k}).
In the following we argue that $\pi_k(\cdot)$s are piece-wise linear discontinuous functions that are hard to fully characterize. We thus cannot follow the approach of deriving $\pi^\ast(\cdot)$ via taking limit of $\pi_k(\cdot), k \geq 0$.
We, however, obtain the optimal policy when the parameters satisfy certain conditions. Subsequently, we propose an approximate policy that is an upper bound on the optimal policy and also equals the optimal policy for a certain range of parameters.}

\color{black}
Observe that $J_0(x) > x^2$ from~\eqref{eqn:j-prime-0}. The first and second terms in the right hand side of~\eqref{eqn:j-prime-k} are greater than the first and second terms, respectively, in the right hand side of~\eqref{eqn:j-prime-0}. So, $J_1(x) >J_0(x)$. Inductively, we can see that $J_k(x)>J_{k-1}(x),\forall~x$. So the sequence $J_k(\cdot)$s converge. We now outline the approach of determining the optimal policy. Let $\pi_k(\cdot)$ be the optimal controls  of the $k$-stage problems~(i.e., optimal controls in~\eqref{eqn:j-prime-0}-\eqref{eqn:j-prime-k}). In the following we argue that $\pi_k(\cdot)$s are piece-wise linear discontinuous functions that are hard to fully characterize. We thus cannot follow the approach of deriving $\pi^\ast(\cdot)$ via taking limit of $\pi_k(\cdot), k \geq 0$. We obtain optimal policy under certain conditions in Proposition~\ref{prop:exact-solution-2-regions}. We also propose an approximate policy $\bar{\pi}(x)$ which forms an upper bound on the optimal policy (see Proposition~\ref{lemma:optimal-approx-0}). We then show that when the parameters does not satisfy the above mentioned conditions $\bar{\pi}(0)=0$, implying $\pi^\ast(0)=0$(see Proposition~\ref{lemma:optimal-approx-0}). So in this region no service is deferred. This way we characterize the optimal policy for all the settings. The detailed analysis follows below.
\color{black}

Let us define $J_{-1}(x) \coloneqq x^2$. We can then unify~\eqref{eqn:j-prime-0} and~\eqref{eqn:j-prime-k}, i.e., we can use~\eqref{eqn:j-prime-k} to describe $J_k(\cdot), k \geq 0$. \color{black}We hypothesize that $J_k(\cdot)$s are quadratic functions and define, for all $k \geq 0$,\color{black}
\begin{equation}
\label{eqn:reference-fixed-cost}
pJ_{k-1}(u)+(1-p)u^2={a}_k u^2+{b}_k u+{c}_k.
\end{equation}
where ${a}_k,{b}_k$ and ${c}_k$ are defined at appropriate places. \color{black}Our hypothesis is clearly true for $k = 0$. In the following we see that it holds for all $k \geq 1$ as well. \color{black}
Also observe that for all $k \geq 0$,
\[
\pi_{k}(x) = \argmin_{u\in [0,\psi]}\{(\psi+x-u)^2 + d+pJ_{k-1}(u)+(1-p)u^2\}
\]
if the minimum value is less than $(\psi+x)^2 + pJ_{k-1}(0)$ and $\pi_{k}(x) = 0$ otherwise. Let us define 
\begin{equation}
\label{eqn:theta-def}
\theta(a,b) \coloneqq \sqrt{d(1+a)} + \frac{b}{2} - \psi \text{ for } a,b \geq 0.  
\end{equation}
We begin with the following observation which we will repeatedly use. \color{black}We use the following lemma later to show that the optimal policy does not defer any service up to certain value of pending service beyond which it defers strictly positive amount.\color{black} 
\begin{lemma}
\label{lma:control-fixed-delay-cost}
\color{black}Let $\pi(x)$ be defined as follows
\begin{equation*}
\pi(x) = 
\begin{cases}
\argmin_{u\in [0,\psi]}\{(\psi+x-u)^2 + d+ au^2 + bu +c\},\\
~~~~~\text{ if }  \min_{u\in [0,\psi]}\{(\psi+x-u)^2 + d+ au^2 + bu +c\} \le (\psi+x)^2+c \\
0,\text{ otherwise}.
\end{cases}
\end{equation*}\color{black}
If $a\psi+\frac{b}{2} \ge \min\{\psi,\theta(a,b)\}$, then
\[
\pi(x) = \begin{cases}
0,&\text{ if } x \le \theta(a,b) \\
\bigg[\frac{x+\psi-\frac{b}{2}}{1+a}\bigg]^\psi,&\text{ otherwise}
\end{cases}
\]
else,
\[
\pi(x) = \begin{cases}
0,&\text{ if } x \le \frac{(a-1)\psi+b}{2}+\frac{d}{2\psi} \\
\psi,&\text{ otherwise}.
\end{cases}
\]
\remove{
If $\frac{2\psi-\frac{b}{2}}{1+a} \le \psi$, then,
\[
\pi(x) = \begin{cases}
0,&\text{ if } x \le x'' \\
\frac{x+\psi-\frac{b}{2}}{1+a},&\text{ otherwise}.
\end{cases}
\]
}
%where, $\tilde{x}=\frac{d+(a-1)\psi^2+b\psi}{2\psi}.$
\end{lemma}
\begin{proof}
See~\ref{appendix:control-fixed-delay-cost}	
%See~\cite[Appendix~I-B]{b29}.
\end{proof}
\begin{remark}
    If $a\psi+\frac{b}{2} \ge \psi$, then $\frac{x+\psi-\frac{b}{2}}{1+a} \le \psi,\forall x \in [0,\psi]$. Therefore,
    \begin{equation}
    \label{eqn:control-fixed-delay-cost}
        \pi(x)=
        \begin{cases}
        0, &\text{ if } x \le \theta(a,b)\\
        \frac{x+\psi-\frac{b}{2}}{1+a}, &\text{otherwise}.
        \end{cases}
    \end{equation}
\end{remark}
Let us define
\begin{equation}
\label{eqn:bar-a-def}
\bar{a}_i = \begin{cases}
1, & \text{if } i = 0, \\
1-\frac{p}{1+\bar{a}_{i-1}}, & \text{otherwise,}
\end{cases}
\end{equation}
and
\begin{equation}
\label{eqn:bar-b-def}
\bar{b}_i = \begin{cases}
2p\psi, & \text{if } i = 0, \\
\frac{p(2\bar{a}_{i-1}\psi+\bar{b}_{i-1})}{1+\bar{a}_{i-1}} & \text{otherwise.}
\end{cases}
\end{equation}
\color{black}We show that the sequences $\bar{a}_k,\bar{b}_k, k \geq 0$ have the following monotonicity properties.We use these properties in deriving the optimal policy (e.g., see the proof of Proposition~\ref{prop:exact-solution-2-regions}).\color{black} 
%We will use the below properties of sequences $\bar{a}_k,\bar{b}_k, k \geq 0$ in the following analysis.
\begin{lemma}
\label{lemma:bar-a-b-convergence}
$(a)$ $\bar{a}_k, k\geq 0$ is a decreasing sequence and converges to $\bar{a}_\infty := \sqrt{1-p}$.\\
	$(b)$ $\bar{b}_k, k\geq 0$ is a decreasing sequence and converges to \\
	$
	\bar{b}_{\infty} := \frac{2p\psi}{1+\sqrt{1-p}}.
	$
\label{lemma:monotonicity-ak-bk-bar}
\end{lemma}
\begin{proof}
%See~\cite[Appendix~I-C]{b29}.
See~\ref{appendix:bar-a-b-convergence}
\end{proof}
\color{black}
The following proposition shows that $\pi^\ast(\cdot)$ is in general a discontinuous piece-wise linear function with increasing slopes. We also know all the affine functions that constitute $\pi^\ast(\cdot)$, but do not know the jump epochs. %In other words optimal policy exhibits discontinuity.
\color{black}
\begin{prop}
\label{prop:pi-star}
The optimal policy $\pi^\ast(\cdot)$ of~\eqref{eqn:fixed-cost-bellman} is of the form%~(a few of the intervals $(\bar{x}_i,\bar{x}_{i+1}]$ can be empty sets)
\begin{equation}
\label{eqn:exact-policy-fixed-delay}
{\pi^*}(x)=
\begin{cases}
0, &\text{ if } 0 \leq x \leq \bar{x}_0  \\
\frac{x+\psi-\frac{\bar{b}_i}{2}}{1+\bar{a}_i}, &\text{  if } \bar{x}_i < x \leq \bar{x}_{i+1}, i\ge 0\\
\frac{x+\psi-\frac{\bar{b}_{\infty}}{2}}{1+\bar{a}_{\infty}}, &\text{  if } \bar{x}_{\infty} < x \leq \psi
\end{cases}
\end{equation}
where $\bar{x}_i,i\ge 0$ are functions of $\bar{a}_i,i\ge 0$ and $\bar{b}_i,i\ge 0$.
\end{prop}
\begin{proof}
%See~\cite[Appendix~I-D]{b29}.
See~\ref{appendix:pi-star}
\end{proof}
 We now provide intuition behind the form of the optimal policies as given by the above proposition. Recall that the  waiting cost $d$ is fixed irrespective of the amount of deferred service whereas the marginal service cost~\cite{roughgarden_2007} in a slot increases with the amount of service offered in the slot. Hence, for optimality, service is deferred only when the marginal service cost in the slot dominates the sum of $d$ and expected marginal service cost in the subsequent slot. Further, given that some service has to be deferred from a slot to the next slot, amount of deferred service is chosen to optimize the service costs in the two slots causing a jump in the optimal policies. Subsequent jumps in the optimal policies can also be attributed to similar phenomenon. Finally, owing to increasing marginal service costs, the optimal policies tend to defer services more aggressively at higher values of pending services. This is why slopes of the successive line segments in the optimal policies increase monotonically. 

We provide the exact optimal policies for a couple of special cases in Section~\ref{subsec:special_cases}. As we do not know the jump epochs in Proposition~\ref{prop:pi-star} we propose an approximate policy in Section~\ref{sec:fixed-waiting-cost-approx}. However, this approximate policy helps us characterize the optimal policy for all cases~(see Section~\ref{sec:fixed-waiting-cost-total-optimal}).

\color{black}\subsection{Optimal policy for Special Cases}\color{black}
\label{subsec:special_cases}
Let us notice that~\eqref{eqn:j-prime-k} for $k \ge 0$ constitute value iteration starting with $J_{-1}(x) = x^2$. We can instead perform value iteration starting with a different function. From~\cite[Chapter~2, Proposition~1.2(b)]{b21}, in any such iteration, $J_k(\cdot)$ will converge to the optimal cost function $J(\cdot)$ and $\pi_k(\cdot)$ will converge to $\pi^\ast(\cdot)$. \color{black}The following proposition shows that, starting with certain initial functions, limits of $\pi_k(\cdot)$ can be obtained in certain special cases.\color{black}

\begin{prop}
\label{prop:exact-solution-2-regions}
$(a)$~If $\psi < \frac{\sqrt{2d}}{(2-p)}$, $\pi^*(x)=0$ for all $x \in [0, \psi]$.\\
$(b)$~If $\psi > \frac{\sqrt{d(1+\bar{a}_{\infty})}}{\bar{a}_{\infty}}$,
\[
\pi^*(x)=\frac{x+\psi-\frac{\bar{b}_{\infty}}{2}}{1+\bar{a}_{\infty}}, \text{ for all }x \in [0, \psi].
\]	
\end{prop}
\begin{proof}
	See~\ref{appendix:exact-solution-2-regions}.
%See Appendix A.
\end{proof}
\subsection{Approximate Policy}
\label{sec:fixed-waiting-cost-approx}
Let us consider a fictitious problem wherein an agent with demand $\psi$ arrives with probability $p$ and there is no arrival with probability $1-p$ but a fixed additional cost $d$ is incurred for each service request whether or not a portion of the request is deferred to the subsequent slot.
The optimal cost function for this fictitious problem, $J'(\cdot)$, is solution of the following Bellman's equation.
\[
J'(x)= \min_{u \in [0, \psi]}
\left\{(\psi-u+x)^2 + d+ pJ'(\psi) + (1-p)u^2 \right\}
\]
This fictitious problem can be seen as a special case of the linear waiting cost problem in~\cite[Section~III]{b28} with $d = 0$ but with a fixed additional cost $d$ per request.
Hence, following the analysis in~\cite[Appendix~C]{b28}~(also see~\cite[Section~3, Theorem~3.1(a)]{b28}), its optimal policy is
\[
\pi'(x)=\frac{x+\psi-\frac{\bar{b}_{\infty}}{2}}{1+\bar{a}_{\infty}},
\]
where $\bar{a}_\infty,\bar{b}_\infty$ are as in Lemma~\ref{lemma:bar-a-b-convergence}. Further, the optimal cost function satisfies
\[
pJ'(x) + (1-p)x^2 = \bar{a}_\infty x^2 + \bar{b}_{\infty} x + \bar{c}_{\infty},
\]
where $\bar{c}_{\infty}$ is a certain constant. Let us now define the following
cost function
\begin{align}
\tilde{J}(x)=\min\left\{(\psi+x)^2 \right.& +pJ'(0),\min_{u\in [0,\psi]}(\psi+x-u)^2 \nonumber \\
&\left. +d+pJ'(u)+(1-p)u^2\right\} \label{eqn:approx-cost}
\end{align}
Also note that
\begin{align}
\label{eqn:a-b-infinity-psi}
\bar{a}_\infty \psi + \frac{\bar{b}_{\infty}}{2}
&= \sqrt{1-p}\psi + \frac{p\psi}{1+\sqrt{1-p}}\nonumber\\
&= \psi.
\end{align}
\color{black}Hence, from~\eqref{eqn:control-fixed-delay-cost}, the optimal control of the cost function $\tilde{J}(x)$, say $\tilde{\pi}(\cdot)$, is given by\color{black}
\begin{equation}
\label{eqn:pi-tilde}
\tilde{\pi}(x)=
\begin{cases}
0, & \text{ if } x \le \theta(\bar{a}_\infty,\bar{b}_\infty)\\
\frac{x+\psi-\frac{\bar{b}_{\infty}}{2}}{1+\bar{a}_{\infty}}, &\text{ otherwise}.
\end{cases}
\end{equation}
We propose to use the following policy for our fixed waiting cost problem.
\begin{equation}
\label{eqn:approx-policy}
\bar{\pi}(x)=
\begin{cases}
0, & \text{ if } \psi < \frac{\sqrt{2d}}{(2-p)}, \\
\frac{x+\psi-\frac{\bar{b}_{\infty}}{2}}{1+\bar{a}_{\infty}}, & \text{ if }
\psi > \frac{\sqrt{d(1+\bar{a}_{\infty})}}{\bar{a}_{\infty}}, \\
\tilde{\pi}(x), & \text{ otherwise.}
\end{cases}
\end{equation}
We do not have any performance bound for the proposed policy. \color{black}However, we show below that, for any given backlog, we defer more under this policy than under the optimal policy.\color{black}
\begin{prop}
\label{lemma:optimal-approx-0}
$\bar{\pi}(x) \geq \pi^*(x)$ for all $x \in [0,\psi]$.
\end{prop}
\begin{proof}
%See~\cite[Appendix~I-F]{b29}.
	See~\ref{appendix:optimal-approx-0}.
\end{proof}
\begin{remark}
	Note that $\bar{\pi}(x)=0$ implies $\pi^*(x)=0$, i.e., the proposed approximate policy and the optimal policy agree when $\bar{\pi}(x)=0$. 
\end{remark}
 
\subsection{Optimal Policy for the general case} 
\label{sec:fixed-waiting-cost-total-optimal}
\color{black}The following theorem completely characterizes the optimal policy.
\begin{theorem}
~\label{thm:opt}
The optimal actions are given as follows %$\pi^\ast(x)$ is given by 
\begin{enumerate}
\item if $\psi > \frac{\sqrt{d(1+\bar{a}_{\infty})}}{\bar{a}_{\infty}}$, then the optimal actions are taken in accordance with $\pi^\ast(x)$ as given by Proposition 3.2(b).
\item $\psi \le \frac{\sqrt{d(1+\bar{a}_{\infty})}}{\bar{a}_{\infty}}$, $\pi^\ast(0)=0$, and therefore none of the requests have their services deferred.%$\pi^\ast(x)=0$.
\end{enumerate}
\end{theorem}
\begin{proof}
Following the definitions of $\bar{a}_{\infty}, \bar{b}_{\infty}$ and $\theta(\bar{a}_{\infty}, \bar{b}_{\infty})$ (see Lemma~\ref{lemma:bar-a-b-convergence} and \eqref{eqn:theta-def}) it can be easily checked that $\theta(\bar{a}_{\infty}, \bar{b}_{\infty})\ge 0$ if and only if $\psi \le \frac{\sqrt{d(1+\bar{a}_{\infty})}}{\bar{a}_{\infty}}$. Hence, if $\psi \le \frac{\sqrt{d(1+\bar{a}_{\infty})}}{\bar{a}_{\infty}}$, $\bar{\pi}(0) = \tilde{\pi}(0) = 0$ from~\eqref{eqn:pi-tilde} and~\eqref{eqn:approx-policy}, and so, $\pi^\ast(0) = 0$ from Proposition 3.3. Notice that when $\pi^\ast(0) = 0$ none of the requests have their services deferred under the optimal policy. 

We thus have complete characterization of the optimal scheduling in all the cases.
\remove{
\begin{enumerate}
\item if $\psi > \frac{\sqrt{d(1+\bar{a}_{\infty})}}{\bar{a}_{\infty}}$, the scheduling decisions are taken in accordance with Proposition 3.2(b).
\item $\psi \le \frac{\sqrt{d(1+\bar{a}_{\infty})}}{\bar{a}_{\infty}}$, none of the requests have their services deferred.
\end{enumerate}}
\end{proof}
\color{black}We illustrate the optimal and the approximate policies via a few examples in Figure~\ref{fig:image1}. We choose $\psi = 2, d= 1$ and  $p = 0.5, 0.7$ and $0.85$ for illustration. When $p = 0.5$, the parameters meet the hypothesis of Proposition~\ref{prop:exact-solution-2-regions}(b), and hence, the optimal policy is provided by the proposition. For $p = 0.7$ and $0.85$, the optimal policies have been computed by value iteration which involves discretization of the state and action spaces and hence is subject to quantization error. For both these cases the approximate policies are given by~\eqref{eqn:approx-policy}. When $p = 0.7$, $\bar{x}_i = \bar{x}_0 > \theta(\bar{a}_{\infty},\bar{b}_{\infty})$ for all $i \geq 1$~(see Proposition~\ref{prop:pi-star}), and hence, the optimal and the approximate policies coincide for $x \geq  \bar{x}_0$. For both, $p = 0.7$ and $0.85$, the optimal policies exhibit jumps and are piece-wise linear with the slopes of successive line segments increasing as claimed in Proposition~\ref{prop:pi-star}).  For both these cases the approximate policies upper bound the optimal policies as shown in Proposition~\ref{lemma:optimal-approx-0}. \textit{As expected, for the same pending service, the deferred service decreases as the expected quantum of service in the next slot increases, i.e., as  $p$ increases.}  
\color{black}

\begin{figure}[!ht]
	
	\centering
	
	\includegraphics[width=0.6\textwidth]{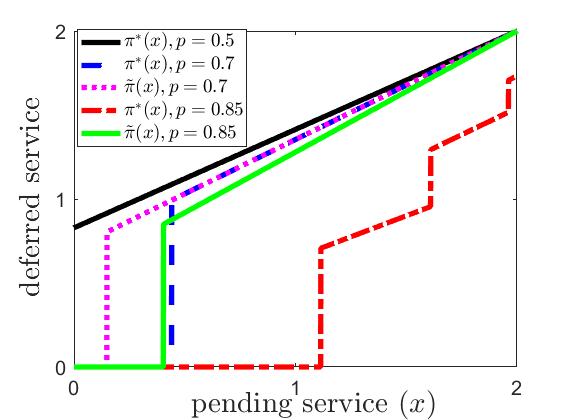}
	\caption{The optimal and the approximate policies for $\psi=2,d=1,p \in \{0.5,0.7,0.85\}$. For $p=0.5$ the optimal and the approximate policies are same.}
	\label{fig:image1}
\end{figure}
\remove{
\begin{figure}[!ht]
	
	\centering
	
	\includegraphics[width=0.29\textwidth]{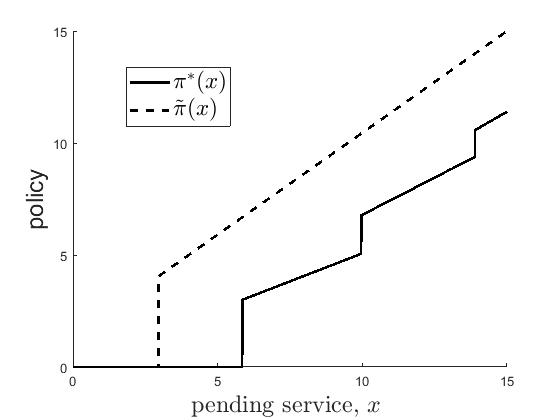}
	\caption{The optimal and the approximate policies for $p=0.99,\psi=15,d=18$.}
	\label{fig:twoslot2}
\end{figure}
\begin{figure}[!ht]
	
	\centering
	
	\includegraphics[width=0.29\textwidth]{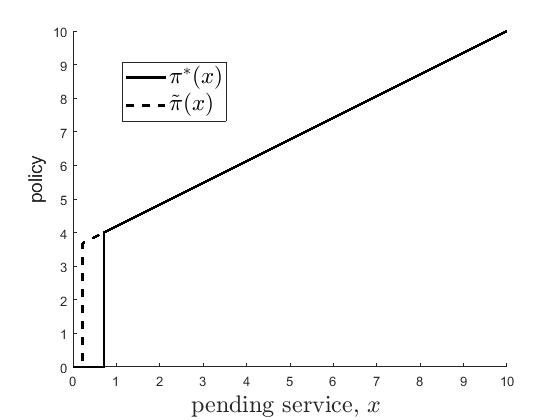}
	\caption{The optimal and the approximate policies for $p=0.7,\psi=10,d=21$.}
	\label{fig:twoslot3}
\end{figure}
}
\vspace{-0.17in}
\subsection*{More general models}We agree that the our model is quite simple and does not capture many attributes of real problems. However, evidently, analysis and optimization of this simple model also is very complex. Further, the optimal solution to this model can lead to heuristics for more general models. We briefly discuss here one such generalization allowing more general demand arrival processes. Assume that, in each slot, with probabilities $p_i$ demands $\psi_i$ arrive where $i = 1,2,\dots,N$, and with probability $1 - \sum_{i=1}^N p_i$ no demand arrives. We can formulate a fictitious problem with i.i.d. Bernoulli arrivals with arrival constant demand $\bar{\psi} = (\sum_{j=1}^Np_j\psi_j)/\bar{p}$ and demand arrival probability $\bar{p} = \sum_{i=1}^N p_i$. We can then use the optimal policy associated with this fictitious problem for our original problem. Such heuristics are proposed and analyzed in~\cite{our-journal} in the context of linear waiting costs. Let us define $\bar{\psi} \coloneqq (\sum_{j=1}^Np_j\psi_j)/\bar{p}$ and consider a fictitious problem wherein an agent with demand $\bar{\psi}$ arrives with probability $\bar{p}$ and there is no arrival with probability $1 - \bar{p}$. Let us further define the following
\[\bar{a}_\infty := \sqrt{1-\bar{p}}\] and \[\bar{b}_{\infty} := \frac{2\bar{p}\bar{\psi}}{1+\sqrt{1-\bar{p}}}.\]
Also, let $\tilde{\pi}(\psi_i,x)$ denote the suggested action for pending work of $x$ units and new arrival of $\psi_i$ units.
Now using this fictitious problem setup and Theorem~3.1 as follows

The optimal actions are given as follows %$\pi^\ast(x)$ is given by 
\begin{enumerate}
\item if $\bar{\psi} > \frac{\sqrt{d(1+\bar{a}_{\infty})}}{\bar{a}_{\infty}}$, then 
\[
\tilde{\pi}(\psi_i,x)=\bigg[\frac{x+\psi_i-\frac{\bar{b}_{\infty}}{2}}{1+\bar{a}_{\infty}}\bigg]^{\psi_i}, \text{ for all }x \in [0, \psi_i], \forall i \in \{1,2,\dots,N\}
\]
\item $\bar{\psi} \le \frac{\sqrt{d(1+\bar{a}_{\infty})}}{\bar{a}_{\infty}}$, $\tilde{\pi}(\psi_i,0)=0$, and therefore none of the requests has their services deferred.%$\pi^\ast(x)=0$.
\end{enumerate} 
\section{NASH EQUILIBRIUM}
\label{sec:nash-equilibrium-fixed}
In this section we provide a Nash equilibrium for the non-cooperative game among the selfish agents~(see Section~\ref{sec:system-model}). Specifically, we look at symmetric Nash equilibria where each agent's strategy is a piece-wise linear function of the remaining demand of the previous player.

%We omit the proofs for lack of space.
Let $C:[0,\psi]\to \mathbb{R}_+$ give the optimal cost for a player as a function of the pending demand given that all other players use strategy $\pi':[0,\psi]\to[0,\psi]$. Clearly, $C(x)$ is given by the following equation for all $x\in[0,\psi]$.
\begin{align*}
&C(x)=\min\{(\psi+x)\psi,\min_{u \in [0,\psi]}\{(\psi-u)(\psi-u+x)+u(u+p(\psi-\pi'(u)))+d\}\}
\end{align*}
We call $\bar{\pi}'=(\pi',\pi',..)$ a symmetric Nash equilibrium if ${\pi}'(x)$ attains the optimal cost in the above optimization problem for all $x $, i.e., if
\[
\pi'(x) = \argmin_{u \in [0,\psi]}\{(\psi-u)(\psi-u+x)+u(u+p(\psi-\pi'(u)))+d\}
\]
if the minimum value is less than $(\psi+x)^2 + c$ and $\pi'(x) = 0$ otherwise, for all $x \in [0,\psi]$. We characterize one such Nash equilibrium in the following. As in section~\ref{sec:fixed-waiting-cost} we define $k$-stage problems, where the tagged player has atmost $k$ service requests after it, before the terminal state is hit. Let $C_k(\cdot)$ be the tagged user's optimal cost in the $k$-stage problem and $\pi'_k(\cdot)$ be the corresponding optimal strategy. Then
\begin{align}
    \label{eqn:nash-c-0-cost-fixed}
C_0(x)=&\min\{(\psi+x)\psi \min_{u \in [0,\psi]}\{(\psi-u)(\psi-u+x)+u^2+d\}\}
\end{align}
and for all $k \ge 1$,
\begin{align}
    \label{eqn:fixed-nash-j-prime-k}
C_k(x)=&\min\{(\psi+x)\psi,  \min_{u \in [0,\psi]}\{(\psi-u)(\psi-u+x)+d+u(u
+p(\psi-\pi'_{k-1}(u)))\}.
\end{align}
We can see $C(x)$ as the limit of $C_k(x)$ as $k$ approaches infinity. Furthermore, the limit of the optimal strategy of $k$-stage problems yield a symmetric Nash equilibrium. \color{black}We now outline the approach of determining the Nash equilibrium, policy. We obtain Nash equilibrium policy under certain conditions in Lemma~\ref{lemma:pi-prime-x-less-than-xinfinity}
and Proposition~\ref{thm:fixed-nash-equilibrium}. Later we characterize total Nash equilibrium policy in Theorem~\ref{thm:nash}.
\color{black}
\remove{
For the ease of notation let us define the following
\begin{equation}
\label{eqn:game-reference-fixed-cost}
\pi'_{k}(x)={a}_k x+{b}_k.
\end{equation}
where ${a}_k,{b}_k$ are defined at appropriate places.
Let us define $\pi_{-1}(x) \coloneqq \psi$. We can then unify~\eqref{eqn:nash-c-0-cost-fixed} and~\eqref{eqn:fixed-nash-j-prime-k}, i.e., we can use~\eqref{eqn:fixed-nash-j-prime-k} to describe $J_k(\cdot), k \geq 0$. It can be observed that $a_{-1}=0,b_{-1}=\psi$}
\subsection{A symmetric Nash equilibrium for special case}
\label{subsec:nash-special-case}
We first focus on symmetric Nash equilirium in a few special cases. We then use these results to obtain symmetric Nash equilibria for all the cases (see Section~\ref{sec:fixed-waiting-cost-total-nash}). We begin with defining sequences
$\tilde{a}_k,\tilde{b}_k,k \ge -1$ as follows
\begin{align}
\tilde{a}_k &= \begin{cases}
0, & \text{if } k = -1 \\
\frac{1}{2(2-p\tilde{a}_{k-1})}, & \text{otherwise}
\end{cases} \label{eqn:fixed-a-case2-nash} \\
\tilde{b}_k &= \begin{cases}
0, & \text{if } k = -1 \\
\frac{(2-p)\psi+p\tilde{b}_{k-1}}{2(2-p\tilde{a}_{k-1})}, & \text{otherwise}
\end{cases} \label{eqn:fixed-b-case2-nash}
\end{align}
We state a few properties of the above sequences.
\begin{lemma}
\label{lemma:fixed-nash-properties-a-b}
%\begin{enumerate}
    $(a)$ The sequence $\tilde{a}_k, k \ge -1$ converges to \[\tilde{a}_{\infty}\coloneqq \frac{1}{p}-\frac{\sqrt{4-2p}}{2p}.\] Also,  $\frac{1}{4}< \tilde{a}_{\infty}<\frac{1}{3}$.\\
   $(b)$ The sequence $\tilde{b}_k, k \ge -1$ converges to \[\tilde{b}_{\infty}\coloneqq \frac{\tilde{a}_{\infty}(2-p)\psi}{1-\tilde{a}_{\infty}p}.\]
%\end{enumerate}
\end{lemma}
\begin{proof}
%See~\cite[Appendix~II-A]{b29}.
See~\ref{appendix:fixed-nash-properties-a-b}.
\end{proof}
\color{black}The following lemma states that $ \tilde{a}_{\infty} x+\tilde{b}_{\infty}$ is strictly positive and strictly less than $\psi$ for all $x \in [0,\psi]$. We use it later to show that under certain conditions, the symmetric Nash equilibria can be obtained via solving unconstrained optimization problems. \color{black}
\begin{lemma}
 $ \tilde{a}_{\infty} x+\tilde{b}_{\infty} \in (0,\psi)$ for all $ x \in [0,\psi]$.
\label{lemma:fixed-game-no-caps}
\end{lemma}
\begin{proof}
%See~\cite[Appendix~II-B]{b29}.
See~\ref{appendix:fixed-game-no-caps}.
\end{proof}
\vspace{0.05in}
Let us also define
$ x_{\infty}=\frac{\sqrt{2\tilde{a}_{\infty}d}-\tilde{b}_{\infty}}{\tilde{a}_{\infty}}. $
%%%$0 < a'_k x+b'_k < \psi$ for all $0 \leq x \leq \psi,k \ge 0$.
%\label{lemma:fixed-game-no-caps}
%\end{lemma}
%\begin{proof}
%See Appendix~\ref{app:fixed-game-no-caps}.
%\end{proof}
%
\color{black}
The following lemma partially characterizes symmetric Nash equilibrium policies. \color{black} 
\begin{lemma}
\label{lemma:pi-prime-x-less-than-xinfinity}
\[{\pi}'(x)=0, \forall x \le x_{\infty}\].
\end{lemma}
\begin{proof}
See~\ref{appendix:pi-prime-x-less-than-xinfinity}.
%See~\cite[Appendix~II-D]{b29}.
\end{proof}
The following proposition gives a symmetric Nash equilibrium in a special case. 
\begin{prop}
\label{thm:fixed-nash-equilibrium}
If $ \frac{\tilde{b}_{\infty}}{1-\tilde{a}_{\infty}} \ge x_{\infty}$, then $\bar{\pi}'=(\pi',\pi',...)$ is a symmetric Nash equilibrium where
\begin{equation}
\label{eqn:fixed-pi-prime}
{\pi}'(x)=
\begin{cases}
0, & \text{ if } x \le x_{\infty} \\
\tilde{a}_{\infty}x+\tilde{b}_{\infty}, &\text{ otherwise}.
\end{cases}
\end{equation}
\end{prop}
\begin{proof}
See~\ref{appendix:fixed-nash-equilibrium}.
%See Appendix B.
\end{proof}
\remove{
\begin{remark}
It has to be noted that $x_{\infty}<0$ implies $\tilde{a}_{\infty}x_{\infty}+\tilde{b}_{\infty}>x_{\infty}$ (as $\tilde{b}_{\infty}>0$,$1-\tilde{a}_{\infty}>0$). Therefore, from Lemma~\ref{lemma:pi-prime-x-less-than-xinfinity} we have complete characterization of Nash equilibrium when the system starts with zero pending.
\end{remark}
}

Notice that the Nash equilibrium as given by Proposition~\ref{thm:fixed-nash-equilibrium} can also have a discontinuity. This jump can be explained using a similar argument as for the jumps in optimal policies~(see the paragraph following Proposition~\ref{prop:pi-star}).

\subsection{Nash equilibrium for the general case} 
\color{black}
The following theorem completely characterizes Nash equilibrium policy.
\label{sec:fixed-waiting-cost-total-nash}
\begin{theorem}
\label{thm:nash}
The Nash equilibrium actions are given as follows %$\pi^\ast(x)$ is given by 
\begin{enumerate}
\item If $x_\infty \geq 0$, then $\pi'(0) = 0$, none of the requests have their services deferred.  
\item If $x_\infty < 0$, then $\pi'(x)$, Nash equilibrium actions are taken in accordance with  Proposition~\ref{thm:fixed-nash-equilibrium}.
\end{enumerate}
\end{theorem}
\begin{proof}
If $x_\infty \geq 0$, $\pi'(0) = 0$ from Lemma~\ref{lemma:pi-prime-x-less-than-xinfinity}. In this case, none of the agents defer any service as they do not see any pending service. On the other hand, if $x_\infty < 0$, Proposition~\ref{thm:fixed-nash-equilibrium} applies, giving the equilibrium scheduling decisions. We thus have complete characterization of the users' scheduling decisions in all the cases.
\end{proof}\color{black}

In Figure~\ref{fig:fixed-image2}, we illustrate symmetric Nash equilibria for the same parameters as used to illustrate the optimal policies in Section~\ref{sec:fixed-waiting-cost}. In all these examples, it turns out that $x_\infty < 0$, and hence, the equilibria are given by Proposition~\ref{thm:fixed-nash-equilibrium}. For the same reason the equilibria do not exhibit jumps. 
\color{black}
\begin{figure}[!ht]
	
	\centering
	
	\includegraphics[width=0.6\textwidth]{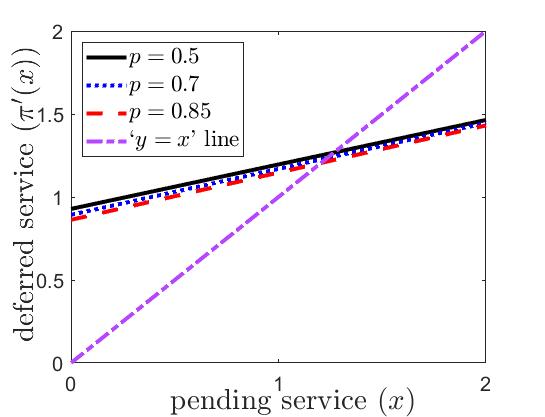}
	\caption{The symmetric Nash equilibria for $\psi=2,d=1,p\in \{0.5,0.7,0.85\}$.}
	\label{fig:fixed-image2}
\end{figure}

\section{Comparative Numerical Evaluation}
We now discuss the effect of the two waiting cost structures, fixed and quadratic, on the scheduling policies, deferred services and costs. For any given cost structure, we also compare the impact of performance criteria~(optimal scheduling vs strategic scheduling by selfish agents).   
%\subsection{Fixed waiting costs}

We begin with revisiting the optimal policies and Nash equilibria in Figures~\ref{fig:image1} and~\ref{fig:fixed-image2}. Recall that we had chosen $\psi = 2, d= 1,$ and  $p = 0.5, 0.7$ and $0.85$. \color{black}\textit{Notice that for the same parameters and pending service, e.g., for $p = 0.85$ and $x= 1$, the optimal policy may not defer any service whereas the Nash equilibrium may differ substantial amount~(larger than $1$). Also, the equilibria are not as sensitive to $p$ as the optimal policies.}\color{black} 

We show histograms of pending services seen by the jobs for both optimal policies and Nash equilibria in Figure~\ref{fig:fixed-image3,4}. We use $p  = 0.5$ and $p=0.85$ for left subfigure and right subfigure respectively. For $p = 0.85$, since $\pi^*(0)=0$, all the jobs see zero pending service under the optimal scheduling policy. When $\pi(0) > 0, (1-p)$ fraction of jobs see $y_0 = 0$ pending service, and for $k \geq 1$, $p^k (1-p)$ fraction of jobs see $y_k = \pi(y_{k-1})$ pending service~($\pi \equiv \pi^\ast$ for an optimal policy whereas $\pi \equiv \pi'$ for a Nash equilibrium). Notice that, for all $k \geq 0$, $y_k$ are upper bounded by, the fixed point of $\pi(x) = x$. \color{black}\textit{For $p =0.85$, under Nash equilibrium the system attains a steady state wherein each user observes a pending service $=1.2053$ (the fixed point of $\pi'(x) = x$ in Figure~\ref{fig:fixed-image2} and defers the same amount of service}\color{black}. Hence we see a mass $(1-p)\sum_{k = 9}^\infty p^k = p^9$ at $y_9 =1.2053$.
%For $p = 0.85$, owing to quantization of the pending services to arrive at a histogram, $y_k \approx 1.2053$ for $k \geq 9$, and hence we see a mass $(1-p)\sum_{k = 9}^\infty p^k = p^9$ at $y_9 =1.2053$.

Next, in Figure~\ref{fig:fixed-image6}(a), we show variation of time-average cost under both optimal policy and Nash equilibrium as $p$ is varied from 0 to 1. In Figure~\ref{fig:fixed-image6}(b), we show price of anarchy vs. $p$. We consider two sets of other parameters, $\psi = 2, d =1$ and $\psi = 2.5, d = 1.5$. For  $\psi = 2, d =1$ and $p \geq 0.58$, no service is deferred in any slot under the optimal policy. Hence, the optimal average cost is $p\psi^2$ in this regime. Under the Nash equilibrium for $p = 1$, the system attains a steady state wherein each user observes a pending service given by the fixed point of $\pi'(x) = x$ and defers the same amount of service. Consequently, the amount of offered service in each slot equals $\psi$ in the steady state, and the average cost equals $\psi^2+d$. \color{black}\textit{The ratio of the average cost under Nash equilibrium and the optimal cost, often termed as efficiency loss, is 1 for $p \gtrsim 0$ and $1 + \frac{d}{\psi^2}$  for $p =1$}\color{black}. We observe same phenomena for $\psi =2.5, d = 1.5$. 

\begin{figure}[!ht]	
\centering
	\includegraphics[width=0.95\textwidth]{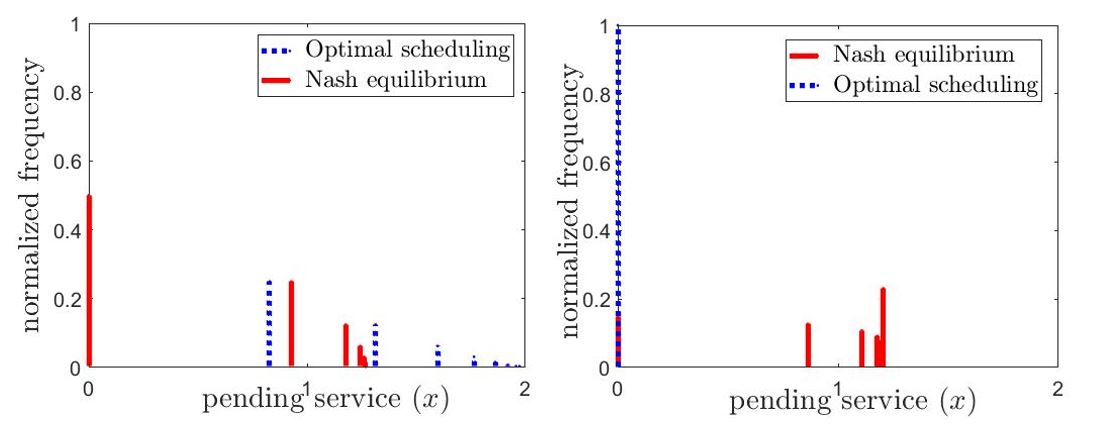}
	\caption{Fixed waiting costs: histogram of the pending services seen by the jobs for $\psi=2,d=1,p=0.5$ (left subfigure) and $p=0.85$ (right subfigure).}
	\label{fig:fixed-image3,4}
\end{figure}
\begin{figure}[!ht]
	\begin{subfigure}{0.5\textwidth}
	\centering
	\includegraphics[width=1\linewidth]{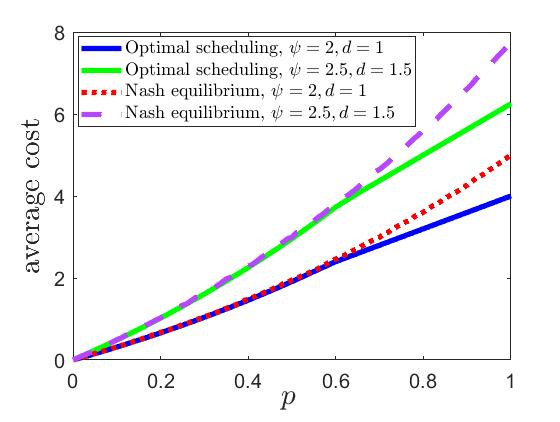}
%		\caption{Average cost vs. $p$ for $\psi=2,d=1$ and \\ ~~~$\psi=2.5,d=1.5$.}
	%	\label{fig:fixed-image5}
	\end{subfigure}
	\begin{subfigure}{0.5\textwidth}
	\centering
	\includegraphics[width=1\linewidth]{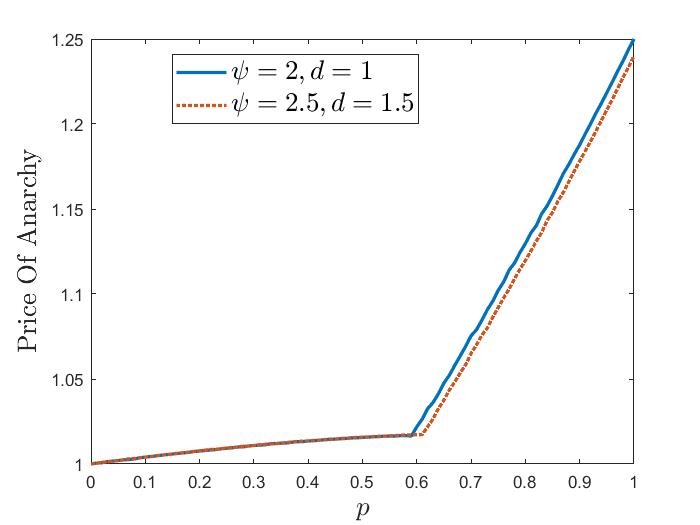}
%	\caption{Price of Anarchy vs. $p$ for $\psi=2,d=1$ and ~~~$\psi=2.5,d=1.5$.}

	\end{subfigure}
	\vspace{-0.1in}
	\caption{(a)Average cost vs. $p$ for $\psi=2,d=1$ and $\psi=2.5,d=1.5$.\\~~~(b)Price of Anarchy vs. $p$ for $\psi=2,d=1$ and $\psi=2.5,d=1.5$.}
			\label{fig:fixed-image6}
\end{figure}
%\begin{subfigures}
\remove{
\begin{figure}[!ht]
	
	\centering
	
	\includegraphics[width=0.65\textwidth]{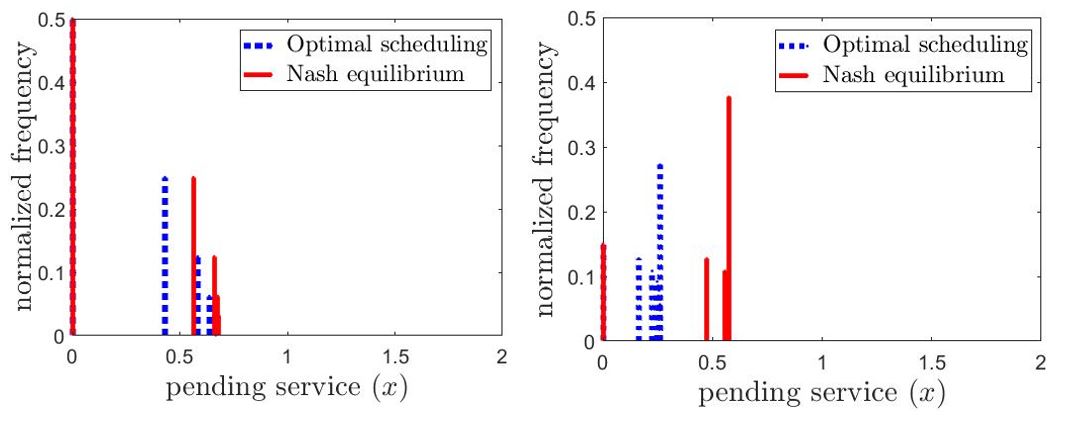}
	\caption{Quadratic waiting costs: histogram of the pending services seen by the jobs for $\psi=2,d=1,p=0.5$(left subfigure) and $p=0.85$(right subfigure).}
	\label{fig:Quadratic-image3,4}
\end{figure}
%\begin{figure}[!ht]
	
%	\centering
	
%	\includegraphics[width=0.29\textwidth]{Qimage4.jpg}
%	\caption{Quadratic waiting costs: histogram of the pending services seen by the jobs for $\psi=2,d=1,p=0.85$.}
%	\label{fig:Quadratic-image4}
%\end{figure}
%\end{subfigures}
\begin{figure}[!ht]
	
	\centering                           
	
	\includegraphics[width=0.4\textwidth]{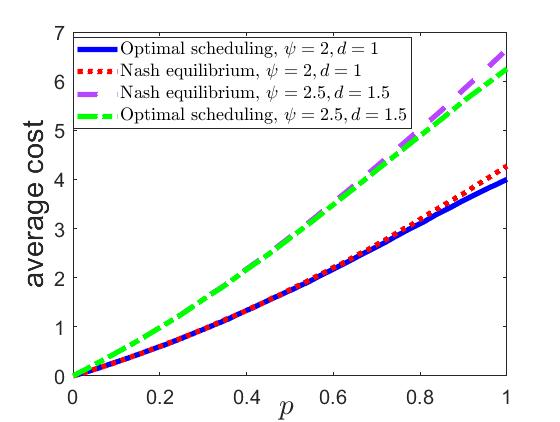}
	\caption{Quadratic waiting costs: average cost vs. $p$ for $\psi=2,d=1$ and $\psi=2.5,d=1.5$.}
	\label{fig:Quadratic-image5}
\end{figure}
%\vspace{-0.4in}
}
\section{Conclusion}
We studied service scheduling in slotted systems with Bernoulli request arrivals, quadratic service costs, fixed and quadratic waiting costs and
service delay guarantee of two slots. In the case of fixed waiting cost, we obtained optimal policy in special cases~(Proposition~\ref{prop:exact-solution-2-regions}). We proposed an approximate policy that is an upper bound on the optimal policy~(Proposition~\ref{lemma:optimal-approx-0}). Finally, we characterize the optimal policy for all cases in Theorem~\ref{thm:opt}. Subsequently, we also provided a symmetric Nash equilibrium when the parameters satisfy certain conditions~(Proposition~\ref{thm:fixed-nash-equilibrium}). And the total characterization of Nash equilibrium can be found in Theorem~\ref{thm:nash}.

\remove{
In case of quadratic waiting costs we provided explicit optimal policy~(Theorem~\ref{thm:optimal-policy}). For competing requests, we derived a symmetric Nash equilibrium~(Theorem~\ref{thm:quad-nash-equilibrium}).}

Our future work entails extending the results to the scenario where service delay guarantee is of three or more slots. We would also like to derive online algorithms for the cases where service request statistics are unknown.
%\subsubsection*{Acknowledgments:}
%\color{black}
%\remove{
\subsubsection*{Acknowledgments}
The first and second authors acknowledge support from Research Fellowships of Visvesvaraya PhD Scheme and INSPIRE Faculty Research Grant (DSTO-1363).
%}
\remove{

\bibliography{main}
}
%\section*{References}
\bibliography{references}

\begin{thebibliography}{10}
\expandafter\ifx\csname url\endcsname\relax
  \def\url#1{\texttt{#1}}\fi
\expandafter\ifx\csname urlprefix\endcsname\relax\def\urlprefix{URL }\fi
\expandafter\ifx\csname href\endcsname\relax
  \def\href#1#2{#2} \def\path#1{#1}\fi

\bibitem{Lin-2011-datacenters}
M.~{Lin}, A.~{Wierman}, L.~L.~H. {Andrew}, E.~{Thereska}, Dynamic right-sizing
  for power-proportional data centers, in: 2011 Proceedings IEEE INFOCOM, 2011,
  pp. 1098--1106.
\newblock \href {http://dx.doi.org/10.1109/INFCOM.2011.5934885}
  {\path{doi:10.1109/INFCOM.2011.5934885}}.

\bibitem{Ren-2012-energy-cloud}
S.~{Ren}, M.~{van der Schaar}, Energy-efficient community cloud for real-time
  stream mining, in: 2012 IEEE 51st IEEE Conference on Decision and Control
  (CDC), 2012, pp. 424--429.
\newblock \href {http://dx.doi.org/10.1109/CDC.2012.6425967}
  {\path{doi:10.1109/CDC.2012.6425967}}.

\bibitem{b15}
P.~{Samadi}, A.~{Mohsenian-Rad}, R.~{Schober}, V.~W.~S. {Wong},
  J.~{Jatskevich}, Optimal real-time pricing algorithm based on utility
  maximization for smart grid, in: 2010 First IEEE International Conference on
  Smart Grid Communications, 2010, pp. 415--420.
\newblock \href {http://dx.doi.org/10.1109/SMARTGRID.2010.5622077}
  {\path{doi:10.1109/SMARTGRID.2010.5622077}}.

\bibitem{b18}
Q.~{Wang}, X.~{Liu}, J.~{Du}, F.~{Kong}, Smart charging for electric vehicles:
  A survey from the algorithmic perspective, IEEE Communications Surveys
  Tutorials 18~(2) (2016) 1500--1517.
\newblock \href {http://dx.doi.org/10.1109/COMST.2016.2518628}
  {\path{doi:10.1109/COMST.2016.2518628}}.

\bibitem{Bae-Kwasinski-12}
S.~{Bae}, A.~{Kwasinski}, Spatial and temporal model of electric vehicle
  charging demand, IEEE Transactions on Smart Grid 3~(1) (2012) 394--403.
\newblock \href {http://dx.doi.org/10.1109/TSG.2011.2159278}
  {\path{doi:10.1109/TSG.2011.2159278}}.

\bibitem{Gusrialdi-et-al-14}
A.~{Gusrialdi}, Z.~{Qu}, M.~A. {Simaan}, Scheduling and cooperative control of
  electric vehicles' charging at highway service stations, in: 53rd IEEE
  Conference on Decision and Control, 2014, pp. 6465--6471.
\newblock \href {http://dx.doi.org/10.1109/CDC.2014.7040403}
  {\path{doi:10.1109/CDC.2014.7040403}}.

\bibitem{b1}
M.~K. Hanawal, E.~Altman, R.~El-Azouzi, B.~J. Prabhu, Spatio-temporal control
  for dynamic routing games, in: R.~Jain, R.~Kannan (Eds.), Game Theory for
  Networks, Springer Berlin Heidelberg, Berlin, Heidelberg, 2012, pp. 205--220.

\bibitem{b10}
R.~Burra, C.~Singh, J.~Kuri, E.~Altman, Routing on a Ring Network, Springer
  International Publishing, Cham, 2019, pp. 25--36.

\bibitem{b2}
F.~{Yao}, A.~{Demers}, S.~{Shenker}, A scheduling model for reduced cpu energy,
  in: Proceedings of IEEE 36th Annual Foundations of Computer Science, 1995,
  pp. 374--382.
\newblock \href {http://dx.doi.org/10.1109/SFCS.1995.492493}
  {\path{doi:10.1109/SFCS.1995.492493}}.

\bibitem{b3}
J.~V. {Gautam}, H.~B. {Prajapati}, V.~K. {Dabhi}, S.~{Chaudhary}, A survey on
  job scheduling algorithms in big data processing, in: 2015 IEEE International
  Conference on Electrical, Computer and Communication Technologies (ICECCT),
  2015, pp. 1--11.
\newblock \href {http://dx.doi.org/10.1109/ICECCT.2015.7226035}
  {\path{doi:10.1109/ICECCT.2015.7226035}}.

\bibitem{b7}
Y.~{Wang}, X.~{Wu}, Y.~{Yu}, W.~{Li}, Manufacturing chain and it's production
  scheduling problem, in: 2007 IEEE International Conference on Control and
  Automation, 2007, pp. 1435--1439.
\newblock \href {http://dx.doi.org/10.1109/ICCA.2007.4376598}
  {\path{doi:10.1109/ICCA.2007.4376598}}.

\bibitem{b26}
G.~Reddy, R.~Vaze, Robust online speed scaling with deadline uncertainty, in:
  Approximation, Randomization, and Combinatorial Optimization. Algorithms and
  Techniques, {APPROX/RANDOM} 2018, August 20-22, 2018 - Princeton, NJ, {USA},
  2018, pp. 22:1--22:17.
\newblock \href {http://dx.doi.org/10.4230/LIPIcs.APPROX-RANDOM.2018.22}
  {\path{doi:10.4230/LIPIcs.APPROX-RANDOM.2018.22}}.

\bibitem{b28}
R.~{Burra}, C.~{Singh}, J.~{Kuri}, Service scheduling for bernoulli requests
  and quadratic cost, in: IEEE INFOCOM 2019 - IEEE Conference on Computer
  Communications, 2019, pp. 2584--2592.
\newblock \href {http://dx.doi.org/10.1109/INFOCOM.2019.8737370}
  {\path{doi:10.1109/INFOCOM.2019.8737370}}.

\bibitem{b27}
R.~Burra, C.~Singh, J.~Kuri, Service scheduling for random requests with
  deadlines and linear waiting costs, IEEE Transactions on Network Science and
  Engineering 8~(3) (2021) 2355--2371.
\newblock \href {http://dx.doi.org/10.1109/TNSE.2021.3091763}
  {\path{doi:10.1109/TNSE.2021.3091763}}.

\bibitem{ref2018}
B.~Legros, O.~Jouini, G.~Koole, \href{https://doi.org/10.1287/opre.2017.1652}{A
  uniformization approach for the dynamic control of queueing systems with
  abandonments}, Oper. Res. 66~(1) (2018) 200–209.
\newblock \href {http://dx.doi.org/10.1287/opre.2017.1652}
  {\path{doi:10.1287/opre.2017.1652}}.
\newline\urlprefix\url{https://doi.org/10.1287/opre.2017.1652}

\bibitem{ref2012}
E.~Hyon, A.~Jean-Marie, Scheduling services in a queuing system with impatience
  and setup costs, The Computer Journal 55~(5) (2012) 553--563.
\newblock \href {http://dx.doi.org/10.1093/comjnl/bxq096}
  {\path{doi:10.1093/comjnl/bxq096}}.

\bibitem{b20}
D.~P. Bertsekas, Dynamic Programming and Optimal Control, Vol. II, 3rd Edition,
  Athena Scientific, 2007.

\bibitem{eitan-shimkin-98}
E.~Altman, N.~Shimkin, {Individual Equilibrium and Learning in Processor
  Sharing Systems}, Operations Research 46~(6) (1998) 776--784.
\newblock \href {http://dx.doi.org/10.1287/opre.46.6.776}
  {\path{doi:10.1287/opre.46.6.776}}.

\bibitem{Data-center-routing}
C.~{Joe-Wong}, I.~{Kamitsos}, S.~{Ha}, Interdatacenter job routing and
  scheduling with variable costs and deadlines, IEEE Transactions on Smart Grid
  6~(6) (2015) 2669--2680.

\bibitem{Edge-Computing-System}
Y.~{Zhang}, X.~{Chen}, Y.~{Chen}, Z.~{Li}, J.~{Huang}, Cost efficient
  scheduling for delay-sensitive tasks in edge computing system, in: 2018 IEEE
  International Conference on Services Computing (SCC), 2018, pp. 73--80.

\bibitem{LangTong_scheduling_commitment}
S.~{Chen}, L.~{Tong}, T.~{He}, Optimal deadline scheduling with commitment, in:
  2011 49th Annual Allerton Conference on Communication, Control, and Computing
  (Allerton), 2011, pp. 111--118.
\newblock \href {http://dx.doi.org/10.1109/Allerton.2011.6120157}
  {\path{doi:10.1109/Allerton.2011.6120157}}.

\bibitem{Soft_deadline}
A.~{Srinivasan}, J.~H. {Anderson}, Efficient scheduling of soft real-time
  applications on multiprocessors, in: 15th Euromicro Conference on Real-Time
  Systems, 2003. Proceedings., 2003, pp. 51--59.
\newblock \href {http://dx.doi.org/10.1109/EMRTS.2003.1212727}
  {\path{doi:10.1109/EMRTS.2003.1212727}}.

\bibitem{He_Linear_perunit_electricity}
Y.~{He}, B.~{Venkatesh}, L.~{Guan}, Optimal scheduling for charging and
  discharging of electric vehicles, IEEE Transactions on Smart Grid 3~(3)
  (2012) 1095--1105.
\newblock \href {http://dx.doi.org/10.1109/TSG.2011.2173507}
  {\path{doi:10.1109/TSG.2011.2173507}}.

\bibitem{renewal-reward-theorem}
S.~M. Ross, Stochastic Processes, 2nd Edition, Wiley, 1996.

\bibitem{roughgarden_2007}
T.~Roughgarden, Routing Games, Cambridge University Press, 2007, p. 461–486.
\newblock \href {http://dx.doi.org/10.1017/CBO9780511800481.020}
  {\path{doi:10.1017/CBO9780511800481.020}}.

\bibitem{b21}
D.~P. Bertsekas, Dynamic Programming and Optimal Control, Vol. I, 3rd Edition,
  Athena Scientific, 2007.

\bibitem{our-journal}
R.~Burra, C.~Singh, J.~Kuri, Service scheduling for random requests with
  deadlines and linear waiting costs, IEEE Transactions on Network Science and
  Engineering 8~(3) (2021) 2355--2371.
\newblock \href {http://dx.doi.org/10.1109/TNSE.2021.3091763}
  {\path{doi:10.1109/TNSE.2021.3091763}}.

\end{thebibliography}


\begin{thebibliography}{00}

\bibitem{b1} M. K. Hanawal, E. Altman, R. El-Azouzi and B. J. Prabhu ``Spatio-temporal control for temporal routing games,'' in 2nd International ICST Conference, GAMENETS 2011, April 2011.
\bibitem{b2} F. Yao, A. Demers and S. Shenker, "A scheduling model for reduced CPU energy," Proceedings of IEEE 36th Annual Foundations of Computer Science, Milwaukee, WI, 1995, pp. 374-382.
\bibitem{b3} M. Senthilkumar and P. Ilango. 2016. A Survey on Job Scheduling in Big Data. Cybern. Inf. Technol., vol. 16, no. 3, pp. 35-51, Sept. 2016.
\bibitem{b4} S. T. Maguluri, R. Srikant and L. Ying, "Stochastic models of load balancing and scheduling in cloud computing clusters," 2012 Proceedings IEEE INFOCOM, Orlando, FL, 2012, pp. 702-710.
%\bibitem{b5} X. Mei, X. Chu, H. Liu, Y. W. Leung and Z. Li, "Energy efficient real-time task scheduling on CPU-GPU hybrid clusters," IEEE INFOCOM 2017 - IEEE Conference on Computer Communications, Atlanta, GA, 2017, pp. 1-9.
%\bibitem{b6} V. Erramilli and S. J. Mason, "Multiple Orders Per Job Compatible Batch Scheduling," in IEEE Transactions on Electronics Packaging Manufacturing, vol. 29, no. 4, pp. 285-296, Oct. 2006.
\bibitem{b7} Y. Wang, X. Wu, Y. Yu and W. Li, "Manufacturing Chain and It's Production Scheduling Problem," 2007 IEEE International Conference on Control and Automation, Guangzhou, 2007, pp. 1435-1439.
\bibitem{b8} J. Luo, L. Rao and X. Liu, "Temporal Load Balancing with Service Delay Guarantees for Data Center Energy Cost Optimization," in IEEE Transactions on Parallel and Distributed Systems, vol. 25, no. 3, pp. 775-784, March 2014.
\bibitem{b9} L. Wang et al., "GreenDCN: A General Framework for Achieving Energy Efficiency in Data Center Networks," in IEEE Journal on Selected Areas in Communications, vol. 32, no. 1, pp. 4-15, January 2014.
 \bibitem{b10} R. Burra, C. Singh, J. Kuri and E. Altman, "Routing on a Ring Network," in Game Theory for Networking Applications, Springer International Publishing, pp.25-36, January 2019.
  \bibitem{b11} K. Choi, W. Lee, R. Soma and M. Pedram, "Dynamic voltage and frequency scaling under a precise energy model considering variable and fixed components of the system power dissipation," IEEE/ACM International Conference on Computer Aided Design, 2004. ICCAD-2004., 2004, pp. 29-34.
     \bibitem{b12} P. Arroba, J. M. Moya, J. L. Ayala, R. Buyya, "Dynamic voltage and frequency scaling-aware dynamic consolidation of virtual machines for energy efficient cloud data centers", Concurrency Comput.: Practice Experience, vol. 29, 2017.
 \bibitem{b13}  K. Jonathan, “Growth in data center electricity use 2005 to 2010,”
Oakland, CA: Analytics Press, 2011.
\bibitem{b14} T. C. Chiu, Y. Y. Shih, A. C. Pang and C. W. Pai, "Optimized Day-Ahead Pricing With Renewable Energy Demand-Side Management for Smart Grids," in IEEE Internet of Things Journal, vol. 4, no. 2, pp. 374-383, April 2017.
    \bibitem{b15}P. Samadi, A. H. Mohsenian-Rad, R. Schober, V. W. S. Wong and J. Jatskevich, "Optimal Real-Time Pricing Algorithm Based on Utility Maximization for Smart Grid," 2010 First IEEE International Conference on Smart Grid Communications, Gaithersburg, MD, 2010, pp. 415-420.
 \bibitem{b16} W. Tang, S. Bi and Y. J. Zhang, "Online Coordinated Charging Decision Algorithm for Electric Vehicles Without Future Information," in IEEE Transactions on Smart Grid, vol. 5, no. 6, pp. 2810-2824, Nov. 2014.
\bibitem{b17} H. Mohsenian-Rad and M. Ghamkhari, "Optimal Charging of Electric Vehicles With Uncertain Departure Times: A Closed-Form Solution," in IEEE Transactions on Smart Grid, vol. 6, no. 2, pp. 940-942, March 2015.
\bibitem{b18} Q. Wang, X. Liu, J. Du and F. Kong, "Smart Charging for Electric Vehicles: A Survey From the Algorithmic Perspective," in IEEE Communications Surveys \& Tutorials, vol. 18, no. 2, pp. 1500-1517, Secondquarter 2016.
\bibitem{b19} W. Shi and V. W. S. Wong, "Real-time vehicle-to-grid control algorithm under price uncertainty," 2011 IEEE International Conference on Smart Grid Communications (SmartGridComm), Brussels, 2011, pp. 261-266.
\bibitem{b20} Dimitri P. Bertsekas. 2007. Dynamic Programming and Optimal Control, Vol. II (3rd ed.). Athena Scientific.
\bibitem{b21} Dimitri P. Bertsekas. 2000. Dynamic Programming and Optimal Control, Vol. I (3rd ed.). Athena Scientific.
\bibitem{b22} T. E. Duncan and B. Pasik-Duncan, “Discrete Time Linear Quadratic Control With Arbitrary Correlated Noise," in IEEE Transactions on Automatic Control, vol. 58, no. 5, pp. 1290-1293, May 2013.
\bibitem{b23} D. P. Kothari, I. J. Nagrath, Modern Power System Analysis, 2003
:McGraw-Hill.
\bibitem{b24}Boyd, S., \& Vandenberghe, L. (2004). Convex Optimization. Cambridge: Cambridge University Press. doi:10.1017/CBO9780511804441
\bibitem{b25}Rockafellar, R. T. (1970). Convex Analysis. Princeton, New Jersey:
Princeton Univ. Press.
\end{thebibliography}
%\bibliographystyle{ieeetr}
%\bibliography{references}
\appendix
%\def\thesubsectiondis{\Roman{subsection}.}
%\newtheorem{apptheorem}{Theorem}[subsection]
%\newtheorem{applemma}[apptheorem]{Lemma}
%\newpage
%\appendices
%\def\thesectiondis{\Roman{section}}
%\newtheorem{apptheorem}{Theorem}[subsection]
%\newtheorem{applemma}[apptheorem]{Lemma}
%\vspace{-0.15in}
%%%%%%%%%%%%%%%%%%%%%%%%%%%%%%%%%%%%%%%%%%%%%%%%%%%%%%%%%%%%%%%%%%%%%%%%%%%%
%\section{}
%\remove{
\section{Proof of Lemma~\ref{lemma:expectation-Ai}}
\label{Appendix:expectation-Ai}
 Recall that a new renewal epoch commences in a slot if the slot has a service request but the preceding slot does not. Hence $A_{i+1} - A_i \geq 2$. Further, $A_{i+1} - A_i = j$ if and only if for some $k \in \{1,1,\dots,j-1\}$, slots $A_i+1,\dots,A_i+k-1$ have service requests,  slots $A_i+k,\dots,A_i+j-1$ have no service requests and slot $,A_i+j$ has a service request. So
Therefore,
 \[
 P(A_{i+1}-A_{i}=j)=\sum_{k=1}^{j-1}p^j(1-p)^{j-k}.
 \]
Hence,
\begin{align*}
     \E(A_{i+1}-A_i)=&\sum_{j=2}^{\infty}j\sum_{k=1}^{j-1}p^k(1-p)^{j-k}\\
     =&\sum_{j=2}^{\infty}j(1-p)^j\sum_{k=1}^{j-1}\bigg(\frac{p}{1-p}\bigg)^k\\
     =&\sum_{j=2}^{\infty}j\frac{p(1-p)}{1-2p}((1-p)^{j-1}-p^{j-1})\\
     =&\frac{1}{1-2p}\bigg[p\sum_{j=2}^{\infty}j(1-p)^j-(1-p)\sum_{j=2}^{\infty}jp^j\bigg]\\
     =&\frac{1}{1-2p}\bigg[\frac{1-p}{p}-\frac{p}{1-p}\bigg]\\
     =&\frac{1}{p(1-p)}.
 \end{align*}
\subsection{Proof of Lemma~\ref{lma:control-fixed-delay-cost}}
 \label{appendix:control-fixed-delay-cost}
Let us define
\[
%\label{eqn:approx-cost}
g(x,u)=(\psi+x-u)^2+\bar{d}+au^2+bu+c
\]
and
\[
\mu(x)=\argmin_{u\in [0,\psi]}g(x,u).
\]
We see that
\[
\mu(x)=\left[\frac{x+\psi-\frac{b}{2}}{1+a}\right]^\psi_0,
\]
and
\[
\pi(x) = \begin{cases}
\mu(x),&\text{ if } g(x,\mu(x)) \le (\psi+x)^2 + c \\
0,&\text{ otherwise}.
\end{cases}
\]
We divide the rest of the proof in the following three cases.
{\it Case (a)~$(1-a)\psi \leq \frac{b}{2} \leq \psi$:} In this case
$\psi - \frac{b}{2} \geq 0$ and $\frac{2\psi-b/2}{1+a} \leq \psi$, and hence
\[
\mu(x)=\frac{x+\psi-\frac{b}{2}}{1+a}.
\]
Further, following simple algebra, we can verify that $g(x,\mu(x)) \le (\psi+x)^2 + c$ if and only if $x \geq \theta(a,b)$, implying that
\[
\pi(x) = \begin{cases}
\frac{x+\psi-\frac{b}{2}}{1+a},&\text{ if } x \geq \theta(a,b)\\
0,&\text{ otherwise}.
\end{cases}
\]
{\it Case (b)~$\frac{b}{2} > \psi$:} Here
\[
\mu(x) = \begin{cases}
\frac{x+\psi-\frac{b}{2}}{1+a},&\text{ if } x \ge \frac{b}{2} - \psi\\
0,&\text{ otherwise}.
\end{cases}
\]
But $\frac{b}{2} - \psi < \theta(a,b)$. Clearly, $\pi(x) = 0$ for all $x <\frac{b}{2} - \psi$. In fact, following similar arguments as in Case~(a),
\[
\pi(x) = \begin{cases}
\frac{x+\psi-\frac{b}{2}}{1+a},&\text{ if } x \geq \theta(a,b)\\
0,&\text{ otherwise}.
\end{cases}
\]
{\it Case (c)~$\frac{b}{2} < (1-a)\psi$:} Now
\[
\mu(x) = \begin{cases}
\frac{x+\psi-\frac{b}{2}}{1+a},&\text{ if } x \leq a\psi+\frac{b}{2}\\
\psi,&\text{ otherwise}.
\end{cases}
\]
Let us divide this case into two subcases. Let us first assume that $a\psi+\frac{b}{2} \geq \theta(a,b)$. In this case, following similar arguments as in Case~(a),
\[
\pi(x) = \begin{cases}
0,&\text{ if } x < \theta(a,b) \\
\frac{x+\psi-\frac{b}{2}}{1+a},&\text{ if } \theta(a,b) \le x \le a\psi+\frac{b}{2} \\
\psi,&\text{ otherwise}.
\end{cases}
\]
Now let us consider that $a\psi+\frac{b}{2} < \theta(a,b)$. In this case,
following similar arguments as in Case~(a), $\pi(x) = 0$ for all $x \leq a\psi + \frac{b}{2}$.   Further, for $x > a\psi + \frac{b}{2}$,
\[
\pi(x) = \begin{cases}
0,&\text{ if } g(x,\psi) \geq (\psi+x)^2 + c \\
\psi,&\text{ otherwise}.
\end{cases}
\]
It can be easily checked that $g(x,\psi) \ge (\psi+x)^2 + c$ if and only if
$x \leq \frac{(a-1)\psi+b}{2}+\frac{d}{2\psi}$. Expectedly,
$\frac{(a-1)\psi+b}{2}+\frac{d}{2\psi} > \theta(a,b)$. Hence,
\[
\pi(x) = \begin{cases}
0,&\text{ if } x \leq  \frac{(a-1)\psi+b}{2}+\frac{d}{2\psi}\\
\psi,&\text{ otherwise}.
\end{cases}
\]
Combining Cases~(a),(b) and (c) yields the desired expressions for $\pi(x)$ in various scenarios.

\remove{
\subsection{Proof of Lemma~\ref{lemma:approx-policy-1}}
 \label{appendix:approx-policy-1}
 Using same format as in~\eqref{eqn:jk-general-form}, $\tilde{J}(x)$ can be written as
 \begin{align}
%\label{eqn:approx-cost}
\tilde{J}(x)=&\min\{(\psi+x)^2+c'_{\infty},\min_{u\in [0,\psi]}(\psi+x-u)^2\\ \nonumber
&+d+a'_{\infty}u^2+b'_{\infty}u+c'_{\infty}\}
\end{align}
Clearly, $\pi'(x)$ given by~\eqref{eqn:pi-prime} is the optimal solution of the second objective. Therefore to find $x$ at which minimizer changes from $0$ to a positive value we look at the following
\begin{equation*}
    (x+\psi)^2+c'_{\infty}<\frac{a'_{\infty}(x+\psi)^2-\frac{{b'_{\infty}}^2}{4}+b'_{\infty}(x+\psi)}{1+a'_{\infty}}+d+c'_{\infty}.
\end{equation*}
Equivalently,
\begin{equation*}
    (x+\psi)^2<-\frac{{{b'_{\infty}}^2}}{4}+b'_{\infty}(x+\psi)+d(1+a'_{\infty}).
\end{equation*}
Or,
\begin{equation}
\label{eqn:x-inf-stage}
    x<\sqrt{d(1+a'_{\infty})}+\frac{b'_{\infty}}{2}-\psi=\sqrt{2d}-\psi
\end{equation}
}
\section{Proof of Lemma~\ref{lemma:bar-a-b-convergence}}
\label{appendix:bar-a-b-convergence}
 $(a)$ Notice the mapping $a \mapsto 1 - \frac{p}{1+a}$ is monotonically increasing.
Further, $\bar{a}_0 = 1$ and $\bar{a}_1 = 1-\frac{p}{2} < \bar{a}_0$.
Therefore the sequence $\bar{a}_k, k\geq0$ is monotonically decreasing.
%It is also nonnegative, and so, lower bounded.
Hence it converges to $a_\infty$, the positive fixed point
of $a = 1 - \frac{p}{1+a}$.\\
  $(b)$  We prove via induction that $\bar{b}_i \bar{b}_i < \bar{b}_{i-1}$ for all $i \ge 1$. Recall that $\bar{b}_0=2p\psi$. Hence,
  from~\eqref{eqn:bar-b-def}, $\bar{b}_1 = p(1+p)\psi$.
Clearly, $\bar{b}_1 < \bar{b}_0$. Now
assume that $\bar{b}_{i} < \bar{b}_{i-1}$ for some $i \geq 1$. This implies~(see~\eqref{eqn:bar-b-def})
\[
\bar{b}_{i-1} > \frac{p(2\psi \bar{a}_{i-1}+\bar{b}_{i-1})}{1+\bar{a}_{i-1}},
\]
or equivalently,
\begin{equation}
\label{eqn:i-1-expression}
\bar{b}_{i-1} > \frac{2p\psi \bar{a}_{i-1}}{1+\bar{a}_{i-1}-p}.
\end{equation}
Hence
\begin{align*}
\bar{b}_i &=  \frac{p(2\psi \bar{a}_{i-1}+\bar{b}_{i-1})}{1+\bar{a}_{i-1}}\\
&>\frac{2p\psi\bar{a}_{i-1}}{1+a_{i-1}}\left(1+\frac{p}{1+\bar{a}_{i-1}-p}\right),\\
&= \frac{2p\psi \bar{a}_{i-1}}{1+\bar{a}_{i-1}-p}\\
&> \frac{2p\psi \bar{a}_{i}}{1+\bar{a}_{i}-p},
\end{align*}
where the first inequality follows from~\eqref{eqn:i-1-expression} and the last one from the fact that $\bar{a}_{i-1}>\bar{a}_i$. The resulting inequality is equivalent to~(again see~\eqref{eqn:bar-b-def})
\[
\bar{b}_i > \frac{p(2\psi \bar{a}_i+\bar{b}_i)}{1+\bar{a}_i} = \bar{b}_{i+1}.
\]
This completes the induction.

Next note that $\bar{b}_i, i\geq0$ are also nonnegative. Therefore, since $\bar{a}_k, k \geq 0$ converge to $\bar{a}_\infty$, $\bar{b}_i, i\geq0$ converge to $b_\infty$, the unique fixed point of
\[
b =\frac{p(2\bar{a}_{\infty}\psi+b)}{1+\bar{a}_{\infty}}.
\]

\subsection{Proof of Proposition~\ref{prop:pi-star}}
\label{appendix:pi-star}
Recall that $J_{-1}(x) = x^2$. Substituting $k=0$ in~\eqref{eqn:reference-fixed-cost}, ${a}_0 = 1$ and ${b}_0 = 0$. Observe that  $a_0\psi+\frac{b_0}{2}=\psi$. Therefore, using~\eqref{eqn:control-fixed-delay-cost},
\[
\pi_0(x)=
\begin{cases}
0, &\text{ if } x \leq \theta(a_0,b_0)\\
\frac{x+\psi-\frac{b_0}{2}}{1+a_0}, &\text{ otherwise }.
\end{cases}
\]
%where $x'' = \sqrt{d(1+a_0)}+\frac{b_0}{2}-\psi$, and so,
Let us define
\[
J_{00}(x)\coloneqq(\psi+x)^2 \text{ and }\\
J_{01}(x)\coloneqq(\psi+x-\pi_0(x))^2+d+\pi_0(x)^2.
\]
So,
\[
J_0(x)=
\begin{cases}
J_{00}(x), &\text{ if } x \leq \theta(a_0,b_0)\\
J_{01}(x), &\text{ otherwise}.
\end{cases}
\]
Note that the function $J_0(x)$ can be one of the two quadratic functions $J_{00}(x),J_{01}(x)$ depending upon $x$. Observe that $\pi_0(x)$ is piece-wise linear but discontinuous with a jump at $\theta(a_0,b_0)$. However, by definition of $\theta(a_0,b_0)$
\[
J_{00}(\theta(a_0,b_0))=J_{01}(\theta(a_0,b_0)).
\]
For all $x<\theta(a_0,b_0)$, $J_{00}(\theta(a_0,b_0))<J_{01}(\theta(a_0,b_0))$ and for all $x>\theta(a_0,b_0)$, $J_{00}(x)>J_{01}(x)$. We define $a_{1,0},b_{1,0},c_{1,0},a_{1,1},b_{1,1}$ and $c_{1,1}$  as follows 
\begin{equation*}
%\label{eqn:reference-fixed-cost}
pJ_{00}(u)+(1-p)u^2={a}_{1,0} u^2+{b}_{1,0} u+{c}_{1,0}.
\end{equation*}
and
\begin{equation*}
%\label{eqn:reference-fixed-cost}
pJ_{01}(u)+(1-p)u^2={a}_{1,1} u^2+{b}_{1,1} u+{c}_{1,1}.
\end{equation*}
%Note that we are using the second term under minimization in~\eqref{eqn:j-prime-k} with $k=1$ to define $\pi''_1(x)$
%Let $\pi'_1(x)$ be defined as
%%\pi'_1(x)= \argmin_{u\in [0,\psi]}\{(\psi+x-u)^2 d+pJ_{0}(u)+(1-p)u^2\}\}. \]
Using~\eqref{eqn:reference-fixed-cost} for $k=1$, we obtain the following,
\begin{align}
\label{eqn:fixed-opt-a-1}
a_1&=
\begin{cases}
a_{1,0}=1, &\text{if } u \leq \theta(a_0,b_0)\\
a_{1,1}=1-\frac{p}{1+a_0}, &\text{otherwise},
\end{cases}
\end{align}
\begin{align}
\label{eqn:fixed-opt-b-1}
b_1&=
\begin{cases}
b_{1,0}=2p\psi, &\text{if } u \leq \theta(a_0,b_0)\\
b_{1,1}=\frac{p(2a_0\psi+b_0)}{1+a_0}, &\text{otherwise},
\end{cases}
\end{align}
\begin{align}
\label{eqn:fixed-opt-c-1}
\text{and } c_1&=
\begin{cases}
c_{1,0}=p\psi^2,&\text{if } u \leq \theta(a_0,b_0)\nonumber\\
c_{1,1}=p\big(\frac{a_0 \psi^2 +b_0\psi- \left(\frac{b_0}{2}\right)^2}{1+a_0}+d\big),&\text{otherwise}.\nonumber
\end{cases}
\end{align}
%As $a_{1,0}\psi+\frac{b_{1,0}}{2}>\psi$ and $a_{1,1}\psi+\frac{b_{1,1}}{2}=\psi$, it can be noted that $\pi_1(x)$ can be either $\frac{x+\psi-\frac{b_{1,0}}{2}}{1+a_{1,0}} \text{ or }
%\frac{x+\psi-\frac{b_{1,1}}{2}}{1+a_{1,1}}\text{ or } 0$.
Let us now define the following fictitious cost function.
\[
J'_1(x)=\min_{u\in [0,\psi]}\{(\psi+x-u)^2+d
+a_1u^2+b_1u+c_1\}\},
\]
Let $\pi'_1(x)$ be the optimal action in the R.H.S . $J'_1(x)$ can be equivalently written as
\begin{align*}
    J'_1(x)=&\min\bigg\{\min_{u\in [0,\theta(a_0,b_0)]}\{(\psi+x-u)^2+d
+a_{1,0}u^2+b_{1,0}u+c_{1,0}\},\\&
\min_{u\in [\theta(a_0,b_0),\psi]}\{(\psi+x-u)^2+d
+a_{1,1}u^2+b_{1,1}u+c_{1,1}\}\bigg\}
\end{align*}
Let us define
\begin{align*}
    h_{a,b,c}(x,u)=(\psi+x-u)^2+d +au^2+b^2+c\\
    \text{and }\pi_{a,b}(x)=\argmin_{u \in [0,\psi]}h_{a,b,c}(x,u)
\end{align*}
$(a)$ If $\pi_{a_{10},b_{10}}(x)<\theta(a_0,b_0)$ and $\pi_{a_{11},b_{11}}(x)<\theta(a_0,b_0)$ then
\begin{align*}
\min_{u\in[0,\psi]}h_{a_{10},b_{10}}(x,u)&=\min_{u \in [0,\theta_(a_0,b_0)]}h_{a_{10},b_{10}}(x,u)\\
&<\min_{u \in [0,\psi]}h_{a_{11},b_{11}}(x,u)\\
&\le\min_{u \in [\theta_(a_0,b_0),\psi]}h_{a_{11},b_{11}}(x,u).\\
\end{align*}
Hence in this case $\pi'_1(x)=\pi_{a_{10},b_{10}}(x)$.\\
$(b)$ If $\pi_{a_{10},b_{10}}(x)\ge\theta(a_0,b_0)$ and $\pi_{a_{11},b_{11}}(x)\ge\theta(a_0,b_0)$ then
\begin{align*}
\min_{u\in[0,\psi]}h_{a_{11},b_{11}}(x,u)&=\min_{u \in [\theta_(a_0,b_0),\psi]}h_{a_{11},b_{11}}(x,u)\\
&<\min_{u \in [0,\psi]}h_{a_{10},b_{10}}(x,u)\\
&\le\min_{u \in [0,\theta_(a_0,b_0)]}h_{a_{10},b_{10}}(x,u).\\
\end{align*}
Hence in this case $\pi'_1(x)=\pi_{a_{11},b_{11}}(x)$.
\remove{
\[\pi_{a_{10},b_{10}}(x)\coloneqq \frac{x+\psi-\frac{b_{1,0}}{2}}{1+a_{1,0}} \text{ and }
\pi_{a_{11},b_{11}}(x)\coloneqq\frac{x+\psi-\frac{b_{1,1}}{2}}{1+a_{1,1}}.\]
As $0<\pi_{a_{10},b_{10}}(x)<\psi$ and $0<\pi_{a_{11},b_{11}}(x)\le \psi, \forall x \in [0,\psi]$ (or)As $a_{1,0}\psi+\frac{b_{1,0}}{2}>\psi$ , $\psi>\frac{b_{1,0}}{2}>\frac{b_{1,1}}{2}$ and $a_{1,1}\psi+\frac{b_{1,1}}{2}=\psi$, it can be noted that $\pi'_1(x)$ can be either $\pi_{a_{10},b_{10}}(x)\text{ or } \pi_{a_{11},b_{11}}(x)$. By~\eqref{eqn:fixed-opt-a-1} and~\eqref{eqn:fixed-opt-b-1} we can see that
\begin{align}
\label{eqn:fixed-pi-prime-general}
\pi'_1(x)=
\begin{cases}
\pi_{a_{10},b_{10}}(x), &\text{ if } \pi_{a_{10},b_{10}}(x)\le \theta(a_0,b_0)\text{ and }\pi_{a_{11},b_{11}}(x)\le \theta(a_0,b_0)\\
\pi_{a_{11},b_{11}}(x), &\text{ if } \pi_{a_{10},b_{10}}(x)> \theta(a_0,b_0)\text{ and }\pi_{a_{11},b_{11}}(x)> \theta(a_0,b_0)
\end{cases}
\end{align}
}
Using~\eqref{eqn:fixed-opt-a-1} and~\eqref{eqn:fixed-opt-b-1}, it can be easily verified that there exist $\tilde{x}<0\text{ and }y<0$ such that $\pi_{a_{10},b_{10}}(\tilde{x})=\pi_{a_{11},b_{11}}(\tilde{x})=y$. Let us look at the figure~\ref{fig:pi-prime-graph} that illustrates the same. Please note that the solid line and dashed line are plots corresponding to $\pi_{a_{10},b_{10}}(x)$ and $\pi_{a_{11},b_{11}}(x)$ respectively. Thicker lines refer to the piece-wise linear discontinuous function $\pi_1'(x)$.
\begin{figure}[!ht]
	
	\centering
	
	\includegraphics[width=0.75\textwidth]{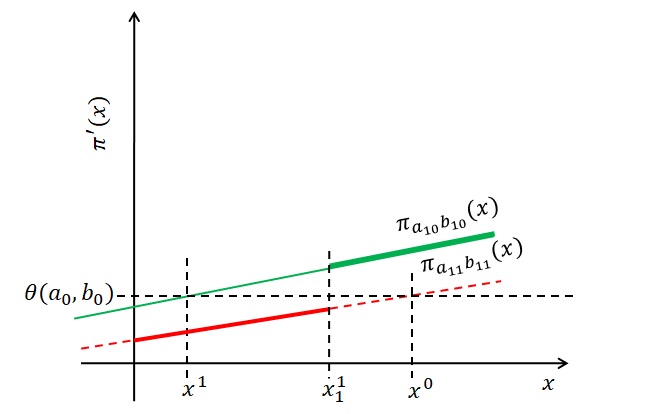}
	\caption{Sample plot of $\pi'(x)$ vs $x$. $x^1,x^0$ in the image are referred to as $x',x''$ in the text}
	\label{fig:pi-prime-graph}
\end{figure}
Let us define $x'$ and $x''$ as follows
\[
x' \coloneqq \max\{x:\pi_{a_{10},b_{10}}(x)<\theta(a_0,b_0) \text{ and }\pi_{a_{11},b_{11}}(x)<\theta(a_0,b_0)\}
\]
\[
x'' \coloneqq \max\{x:\pi_{a_{10},b_{10}}(x)\ge\theta(a_0,b_0) \text{ and } \pi_{a_{11},b_{11}}(x)\ge\theta(a_0,b_0)\}
\]
 Using figure~\ref{fig:pi-prime-graph}, case$(a)$ and case$(b)$, we can see that $\pi'(x)$ can be written as
\begin{align}
\label{eqn:fixed-pi-prime-partial}
\pi'_1(x)=
\begin{cases}
\pi_{a_{10},b_{10}}(x), &\text{ if } x<x'\\
\pi_{a_{11},b_{11}}(x), &\text{ if } x> x''
\end{cases}
\end{align}
Using~\eqref{eqn:fixed-pi-prime-partial} we can write the following
\begin{align}
\label{eqn:fixed-J-prime-partial}
J'_1(x)=
\begin{cases}
\frac{a_{1,0}}{1+a_{1,0}}x^2+(\frac{2\psi a_{1,0}+b_{1,0}}{1+a_{1,0}})x
+\frac{a_{1,0}\psi^2+b_{1,0}\psi-\frac{b_{1,0}^2}{4}}{(1+a_{1,0})}\\+c_{1,0}+d \eqqcolon
A_0x^2+B_0x+C_0, \text{ if } x<x'\\
\frac{a_{1,1}}{1+a_{1,1}}x^2+(\frac{2\psi a_{1,1}+b_{1,1}}{1+a_{1,1}})x
+\frac{a_{1,1}\psi^2+b_{1,1}\psi-\frac{b_{1,1}^2}{4}}{(1+a_{1,1})}\\+c_{1,1}+d\eqqcolon A_1x^2+B_1x+C_1, \text{ if } x> x''.
\end{cases}
\end{align}
Also,
 \begin{align}
\label{eqn:fixed-pi-prime-conditional}
\pi'_1(x)=
\begin{cases}
\pi_{a_{10},b_{10}}(x), &\text{ if } x' \ge \psi\\
\pi_{a_{11},b_{11}}(x), &\text{ if }  x'' \le 0
\end{cases}
\end{align}
Therefore, we only discuss the case $x' > 0$ and $x''<\psi$\footnote{When $x'\ge \psi$, $[\bar{x}_1,\psi]$ is an empty set. Similarly, when $x''\le 0$, then $[\bar{x}_0,\bar{x}_1]$ is an empty set.}. Let us observe that solution to the following equation gives us $\bar{x}_1 \in [x',x''] $
\begin{align}
\label{eqn:comparision-roots}
A_0x^2+B_0x+C_0=A_1x^2+B_1x+C_1
\end{align}
As $x'>0$, $C_0<C_1$. Also, $A_0>A_1$ as $a_{1,0}>a_{1,1}$. Thus the product of roots of~\eqref{eqn:comparision-roots} is negative. Hence there exists a $\bar{x}_1 \in [x',x'']$ such that
\[
\pi'_1(x)=
\begin{cases}
\pi_{a_{10},b_{10}}(x), &\text{ if } 0 < x
\leq \bar{x}_1\\
\pi_{a_{11},b_{11}}(x), &\text{ if } \bar{x}_1 < x \leq \psi
\end{cases}
\]
It can be noted that $\bar{x}_1$ is a function of $a_{1,0},a_{1,1},b_{1,0},b_{1,1},c_{1,0}$ and $c_{1,1}$ but not easy to determine. %\footnote{Here we assume that $0 \leq x^1_1 < \psi$. We may have a different expression for a different ordering. For instance, if $0 > x^1_1$, we may have
%\[
%\pi'_1(x)=\bigg[\frac{x+\psi-\frac{b_{1,1}}{2}}{1+a_{1,1}}\bigg]_0^\psi
%\]	
%where $0 \leq x^1_1 < \psi$.}	
 Further, we study the optimal control for $1$-stage problem.
\begin{align*}
J_1(x)=\min\{(\psi+x)^2+pJ_{0}(0), \min_{u\in [0,\psi]}\{(\psi&+x-u)^2+d\\
+a_1u^2&+b_1u^2+c_1\}\}.
\end{align*}
\remove{
where
\begin{align*}
a_1&=
\begin{cases}
a_{1,0}=1, &\text{if } u \leq x''\\
a_{1,1}=1-\frac{p}{1+a_0}, &\text{otherwise},
\end{cases} \\
b_1&=
\begin{cases}
b_{1,0}=2p\psi, &\text{if } u \leq x''\\
b_{1,1}=\frac{p(2a_0\psi+b_0)}{1+a_0}, &\text{otherwise},
\end{cases}\\
\text{and } c_1&=
\begin{cases}
c_{1,0}=p\psi^2, &\text{if } u \leq x''\\
c_{1,1}=p\left(\frac{a_0 \psi^2 + \left(\frac{b_0}{2}\right)^2}{1+a_0}+d\right), &\text{otherwise}.
\end{cases}
\end{align*}

Note that the expressions for $a_1$ and $b_1$ can be deduced by
}
Also $pJ_{0}(0) = c_{1,0}$ if $0 \leq x''$ and $c_{1,1}$ otherwise. As we are discussing a case where $x' > 0$ and $x''<\psi$, $pJ_{0}(0) = c_{1,0}$. Let us consider the following fictitious cost functions.
\begin{align*}
J_{1,0}(x)=\min\big\{(\psi+x)^2&+c_{1,0}, \min_{u\in [0,\theta(a_0,b_0)]}\{(\psi+x-u)^2\\& +d
+a_{1,0}u^2+b_{1,0}u^2+c_{1,0}\}\big\}
\end{align*}
and
\begin{align*}
J_{1,1}(x)=\min\big\{(\psi+x)^2&+c_{1,0}, \min_{u\in [\theta(a_0,b_0),\psi]}\{(\psi+x-u)^2\\ &+d
+a_{1,1}u^2+b_{1,1}u^2+c_{1,1}\}\big\}.
\end{align*}
Let $\pi_{1,0}(x)\text{ and }\pi_{1,1}(x)$ be the optimal functions of $J_{1,0}(x)\text{ and }J_{1,1}(x)$ respectively. It can be noted that $a_{1,0}\psi+b_{1,0}=\psi$ and $a_{1,1}\psi+b_{1,1}=\psi$. Therefore using~\eqref{eqn:control-fixed-delay-cost},
\begin{equation}
\label{eqn:pi-10}
\pi_{1,0}(x)=
\begin{cases}
0, &\text{ if } 0\leq x \leq \theta(a_{1,0},b_{1,0})\\
\pi_{a_{10},b_{10}}(x), &\text{ if } \theta(a_{1,0},b_{1,0}) < x \leq \psi
\end{cases}
\end{equation}
and
\begin{equation}
\label{eqn:pi-11}
\pi_{1,1}(x)=
\begin{cases}
0, &\text{ if } 0\leq x \leq \theta(a_{1,1},b_{1,1})\\
\pi_{a_{11},b_{11}}(x), &\text{ if } \theta(a_{1,1},b_{1,1}) < x \leq \psi
\end{cases}
\end{equation}
From~\eqref{eqn:pi-10},~\eqref{eqn:pi-11} we see that optimal function of $\pi_1(x)$ can be written as
\[
\pi_1(x)=
\begin{cases}
0, &\text{ if } 0\leq x \leq \bar{x}_0\\
\frac{x+\psi-\frac{b_{1,0}}{2}}{1+a_{1,0}}, &\text{ if } \bar{x}_0 < x
\leq \bar{x}_1\\
\frac{x+\psi-\frac{b_{1,1}}{2}}{1+a_{1,1}}, &\text{ if } \bar{x}_1 < x \leq \psi
\end{cases}
\]
where $\bar{x}_0=\min\{\theta(a_{1,0},b_{1,0}),\theta(a_{1,1},b_{1,1})\}$. We can similarly argue that the optimal policy $\pi^\ast(\cdot)$ is of the form~(a few of the intervals $(\bar{x}_i,\bar{x}_{i+1}]$ can be empty sets)
\begin{equation*}
%\label{eqn:exact-policy-fixed-delay}
{\pi^*}(x)=
\begin{cases}
0, &\text{ if } 0 \leq x \leq \bar{x}_0  \\
\frac{x+\psi-\frac{\bar{b}_i}{2}}{1+\bar{a}_i}, &\text{  if } \bar{x}_i < x \leq \bar{x}_{i+1}, i \ge 0.
\end{cases}
\end{equation*}
\section{Proof of Proposition~\ref{prop:exact-solution-2-regions}}
\label{appendix:exact-solution-2-regions}
 $(a)$ Let us analyze value iteration starting with function $J_0(x) = (x+\psi)^2$. Substituting $k=1$ in~\eqref{eqn:reference-fixed-cost},  ${a}_1=\bar{a}_0$,  ${b}_1=\bar{b}_0$ and $c_1=\psi^2$. Following~\eqref{eqn:j-prime-k},
\begin{align*}
%\label{eqn:j-prime-k}
    J_1(x)&=\min\{(\psi+x)^2+{c}_1, \\
    &\min_{u\in [0,\psi]}(\psi+x-u)^2+{d}+\bar{a}_0u^2+\bar{b}_0 u+{c}_1\}
\end{align*}
Observe that $\bar{a}_0\psi+\frac{\bar{b}_0}{2}>\psi$. Hence, from~\eqref{eqn:control-fixed-delay-cost}, the optimal control in the $1$-stage problem, $\pi_1(x)$, can be written as
\begin{equation*}
    \pi_1(x)=
\begin{cases}
0,&\text{ if } x \le \theta(\bar{a}_0,\bar{b}_0)=\sqrt{2{d}}+p\psi-\psi \\
\frac{x+\psi-p\psi}{2},&\text{ otherwise.}
\end{cases}
\end{equation*}
Note that when $\sqrt{2{d}} \ge (2-p)\psi$, the second case does not arise, i.e., $\pi_1(x)=0$ for all $x\in[0,\psi]$. It implies that
\[
J_1(x)=(x+\psi)^2+p\psi^2
\]
for all $x \in [0,\psi]$. Again using~\eqref{eqn:reference-fixed-cost} for $k=2$, we see that ${a}_2={a}_1,{b}_2={b}_1$. Hence, following similar arguments as before, $\pi_2(x)=0$ for all $x \in [0,\psi]$. Continuing in this fashion we see that for all $k \geq 1$, $\pi_k(x)= 0$ for all $x \in [0,\psi]$. Therefore
$\pi^*(x)=0$ $x \in [0,\psi]$.\\
$(b)$
%\subsection{Proof of Proposition~\ref{prop:pi-star-non-zero}}
%\label{appendix:pi-star-non-zero}
Now we analyze value iteration starting with function $J_0(u)$ that satisfies  \[pJ_0(u)+(1-p)u^2 = \bar{a}_{\infty}u^2+\bar{b}_{\infty}u+\bar{c}_{\infty},\]
where $\bar{a}_{\infty},\bar{b}_{\infty}$ are as defined in Lemma~\ref{lemma:bar-a-b-convergence} and $\bar{c}_{\infty}$ is a certain constant. Substituting $k=1$ in~\eqref{eqn:reference-fixed-cost},  ${a}_1=\bar{a}_{\infty}$ and ${b}_1=\bar{b}_{\infty}$. Following~\eqref{eqn:j-prime-k},
\begin{align*}
%\label{eqn:j-prime-k}
    J_1(x)&=\min\{(\psi+x)^2+\bar{c}_{\infty}, \\
    &\min_{u\in [0,\psi]}(\psi+x-u)^2+{d}+\bar{a}_{\infty} u^2+\bar{b}_{\infty} u+\bar{c}_{\infty}\}.
\end{align*}
Using definitions of $\bar{a}_{\infty}$ and $\bar{b}_{\infty}$, $\bar{a}_\infty\psi+\frac{\bar{b}_\infty}{2}=\psi$. Hence, from~\eqref{eqn:control-fixed-delay-cost},
\begin{equation*}
    \pi_1(x)=
\begin{cases}
 0,&\text{ if } x \le \theta(\bar{a}_\infty,\bar{b}_\infty) \\
\frac{x+\psi-\frac{\bar{b}_{\infty}}{2}}{1+\bar{a}_{\infty}},&\text{ otherwise.}
\end{cases}
\end{equation*}
Further, when $\psi \bar{a}_{\infty} > \sqrt{{d}(1+a_{\infty})}$,
\begin{align*}
\theta(\bar{a}_\infty,\bar{b}_\infty) &= \sqrt{{d}(1+\bar{a}_{\infty})}+\frac{\bar{b}_{\infty}}{2}-\psi\\ &< \frac{\bar{b}_{\infty}}{2} - \psi(1-\bar{a}_\infty)= 0,
\end{align*}
implying that $\pi_1(x)=\frac{x+\psi-\frac{\bar{b}_{\infty}}{2}}{1+\bar{a}_{\infty}}$ for all $x\in[0,\psi]$. It further implies that
\[
J_1(x)=(\psi+x-\pi_1(x))^2+{d}+\bar{a}_{\infty} \pi_1(x)^2+\bar{b}_{\infty} \pi_1(x)+\bar{c}_{\infty}
\]
for all $x\in[0,\psi]$. Again using~\eqref{eqn:reference-fixed-cost} for $k=2$, we see that
\[
{a}_2=1-\frac{p}{1+\bar{a}_{\infty}} \text{ and } {b}_2=\frac{p(2\bar{a}_{\infty}\psi+\bar{b}_{\infty})}{1+\bar{a}_{\infty}}.
\]
Following Lemma~\ref{lemma:bar-a-b-convergence}, ${a}_2=\bar{a}_{\infty}$ and ${b}_2=\bar{b}_{\infty}$. Hence, following similar arguments as before, $\pi_2(x) =\pi_1(x)$ for all $x \in [0,\psi]$. Continuing in this fashion we see that for all $k \ge 1$,  $\pi_k(x) =\pi_1(x)$ for all $x \in [0,\psi]$. Therefore $\pi^\ast(x)=\frac{x+\psi-\frac{\bar{b}_{\infty}}{2}}{1+\bar{a}_{\infty}}$ for all $x \in [0,\psi]$.
\remove{
\begin{equation*}
 % \label{eqn:pi-two-prime}
    \pi_2(x)=
\begin{cases}
& 0,\text{ if } x \le \sqrt{\bar{d}(1+\bar{a}_2)}+\frac{\bar{b}_2}{2}-\psi \\
& \frac{x+\psi-\frac{\bar{b}_2}{2}}{1+\bar{a}_2},\text{ otherwise}\\
\end{cases}
\end{equation*}
Clearly, when $\psi \bar{a}_{\infty} > \sqrt{\bar{d}(1+a_{\infty})}$,
\[
\pi_2(x)=\frac{x+\psi-\frac{\bar{b}_2}{2}}{1+\bar{a}_2}
\]
Now, a similar argument can be extended and it can be shown that when $\psi \bar{a}_{\infty} > \sqrt{\bar{d}(1+a_{\infty})}$,
\[
\pi_k(x)=\frac{x+\psi-\frac{\bar{b}_k}{2}}{1+\bar{a}_k}
\]
where $\bar{a}_k=\bar{a}_{\infty},\bar{b}_k=\bar{b}_{\infty}$. }
%\vspace{-0.2in}
%\remove{
\section{Proof of Proposition~\ref{lemma:optimal-approx-0}}
\label{appendix:optimal-approx-0}
Following Proposition~\ref{prop:exact-solution-2-regions} and~\eqref{eqn:approx-policy} we see that  $\bar{\pi}(x)$ is either $\pi^\ast(x)$ or $\tilde{\pi}(x)$ depending on the parameters. Therefore, it is enough to argue that
\[
\tilde{\pi}(x) \geq \pi^\ast(x)~\forall x \in [0,\psi]
\]
irrespective of the parameters. We prove this by considering the following two cases separately.

{\it Case 1) $x \le  \theta(\bar{a}_\infty,\bar{b}_\infty)$:}
We assume $\theta(\bar{a}_\infty,\bar{b}_\infty) \ge 0$ else this case does not arise. In this case, $\tilde{\pi}(x) = 0$. We will argue that $\pi^\ast(x)$ also equals zero in this case. We will do this via iteratively  showing that $\pi_k(x)  = 0$ for all $k \geq 0$. First recall that $\bar{a}_\infty\psi+\frac{\bar{b}_\infty}{2}=\psi$~(see Section~\ref{sec:fixed-waiting-cost-approx},~\eqref{eqn:a-b-infinity-psi}). From Lemma~\ref{lemma:monotonicity-ak-bk-bar}, $\bar{a}_k \ge \bar{a}_{\infty}$ and $\bar{b}_k \ge \bar{b}_{\infty}$ for all $k \geq 0$. Hence   $\bar{a}_k\psi+\frac{\bar{b}_k}{2}>\psi$ for all $k \geq 0$ and also,
$\theta(\bar{a}_k, \bar{b}_k) > \theta(\bar{a}_\infty, \bar{b}_\infty)$ for all $k \geq 0$.

Let us now consider value iteration starting with function $J_0(x) = (x+\psi)^2$ as in the proof of Proposition~\ref{prop:exact-solution-2-regions}$(a)$. Recall that ${a}_1=1=\bar{a}_0$, ${b}_1=2p\psi=\bar{b}_0$ and
\begin{equation*}
\pi_1(x)=
\begin{cases}
0,&\text{ if } x \le \theta(\bar{a}_0,\bar{b}_0) \\
\frac{x+\psi-p\psi}{2},&\text{ otherwise.}
\end{cases}
\end{equation*}
Clearly, $\pi_1(x) = 0$ for all $x \leq \theta(\bar{a}_\infty,\bar{b}_\infty)$.
Next we analyze $\pi_2(x)$. Using~\eqref{eqn:reference-fixed-cost} for $k=2$,
\[
pJ_1(u)+(1-p)u^2={a}_{2}u^2+{b}_{2}u+{c}_{2},
\]
where
\begin{align*}
 {a}_2 &=
\begin{cases}
{a}_{21}=\bar{a}_0,&\text{ if } u \le \theta(\bar{a}_0,\bar{b}_0) \\
{a}_{22}=1-\frac{p}{1+\bar{a}_0},&\text{ otherwise}
\end{cases}\\
{b}_2 &=
\begin{cases}
{b}_{21}=\bar{b}_0,&\text{ if } u \le \theta(\bar{a}_0,\bar{b}_0) \\
{b}_{22}=\frac{p(2\bar{a}_0\psi+\bar{b}_0)}{1+\bar{a}_0},&\text{ otherwise}
\end{cases}\\
{c}_2&=
\begin{cases}
{c}_{21}=p(\psi^2+\bar{c}_1),&\text{ if } u \le \theta(\bar{a}_0,\bar{b}_0) \\
{c}_{22}=p(\frac{\bar{a}_0\psi^2+\bar{b}_0\psi-\frac{\bar{b}_0^2}{4}}{1+\bar{a}_0}+\bar{c}_1+d),&\text{ otherwise.}
\end{cases}
\end{align*}
\remove{
If $\theta(\bar{a}_1,\bar{b}_1) \ge \psi$, the second case does not arise, i.e., ${a}_2  = \bar{a}_1$ and ${b}_2  = \bar{b}_1$ for all $u \in [0,\psi]$. In this case,
\begin{align*}
J_2(x) &=\min\Big\{(\psi+x)^2+{c}_{21},\\
&\min_{u\in [0,\psi]}(\psi+x-u)^2+{d}+\bar{a}_{1}u^2+\bar{b}_{1}u+{c}_{21}\Big\},
\end{align*}
and hence, $\pi_2(x) = \pi_1(x) = 0$ for all $x \le  \theta(\bar{a}_\infty,\bar{b}_\infty)$. Now, let $\theta(\bar{a}_1,\bar{b}_1) < \psi$. }Note that
\[
{a}_{21}u^2+{b}_{21}u+{c}_{21} < {a}_{22}u^2+{b}_{22}u+{c}_{22}
\]
for all $u \in [0,\theta(\bar{a}_0,\bar{b}_0))$, implying that ${c}_{21} < {c}_{22}$. Moreover,
\begin{align*}
J_2(x) &=\min\Big\{(\psi+x)^2+{c}_{21},\min_{u\in [0,\theta(\bar{a}_{0},\bar{b}_{0})]}(\psi+x-u)^2+{d}+{a}_{21}u^2+{b}_{21}u+{c}_{21},\\
&\min_{u\in [\theta(\bar{a}_{0},\bar{b}_{0}),\psi]}(\psi+x-u)^2+{d}+{a}_{22}u^2+{b}_{22}u+{c}_{22}\Big\}.
\end{align*}
Let us define functions
\begin{align*}
J_{21}(x)&=\min\{(\psi+x)^2+{c}_{21},\min_{u\in [0,\psi]}(\psi+x-u)^2+{d}+{a}_{21}u^2+{b}_{21}u+{c}_{21}\}
\end{align*}
and
\begin{align*}
J_{22}(x)&=\min\{(\psi+x)^2+{c}_{21},\min_{u\in [0,\psi]}(\psi+x-u)^2+{d}+{a}_{22}u^2+{b}_{22}u+{c}_{22}\}.
\end{align*}
The optimal controls in the above optimization problems are
\begin{equation*}
{\pi_{21}}(x)=
\begin{cases}
0,& \text{ if } x \le \theta({a}_{21},{b}_{21})\\
\frac{x+\psi-\frac{{b}_{21}}{2}}{1+{a}_{21}},& \text{ otherwise}
\end{cases}
\end{equation*}
and
\vspace{-0.1in}
\begin{equation*}
{\pi_{22}}(x)=
\begin{cases}
0,& \text{ if } x \le \sqrt{({d}+{c}_{22}-{c}_{21})(1+{a}_{22})}\\
&~~~~~~~~~~+\frac{{b}_{22}}{2}-\psi \\
\frac{x+\psi-\frac{{b}_{22}}{2}}{1+{a}_{22}},& \text{ otherwise}
\end{cases}
\end{equation*}
respectively. Note that,  since ${c}_{22} > {c}_{21}$,  $\sqrt{({d}+{c}_{22}-{c}_{21})(1+{a}_{22})}+\frac{{b}_{22}}{2}-\psi > \theta({a}_{22},b_{22})$,  and hence, $\pi_{22}(x) = 0$ for all $x \in [0, \theta({a}_{22},{b}_{22})]$. Finally, comparing $J_2$, $J_{21}$ and $J_{22}$, we see that when both $\pi_{21}(x)$ and $\pi_{22}(x)$ equal zero, $\pi_2(x)$ also equals zero. In other words, $\pi_2(x) = 0$ for all $x \le \min\{\theta({a}_{21},b_{21}),\theta({a}_{22},{b}_{22})\}$. In particular, $\pi_2(x) = 0$ for all $x \le \theta(\bar{a}_\infty,\bar{b}_\infty)$.

We can similarly argue that, for all $k \ge 1$, $\pi_k(x) = 0$ for all $x \le \theta(\bar{a}_\infty,\bar{b}_\infty)$ as desired.\\
{\it Case 2) $x >  \theta(\bar{a}_\infty,\bar{b}_\infty)$:}
In this case
\[
\tilde{\pi}(x)=\frac{x+\psi-\frac{\bar{b}_{\infty}}{2}}{1+\bar{a}_{\infty}}.
\]
From Lemma~\ref{lemma:monotonicity-ak-bk-bar}, $\bar{a}_k \ge \bar{a}_{\infty}$ and $\bar{b}_k \ge \bar{b}_{\infty}$ for all $k \geq 0$, and hence,
\[
\tilde{\pi}(x) \ge \frac{x+\psi-\frac{\bar{b}_k}{2}}{1+\bar{a}_k}
\]
for all $x >  \theta(\bar{a}_\infty,\bar{b}_\infty)$. Therefore, following~\eqref{eqn:exact-policy-fixed-delay},
$\tilde{\pi}(x) \ge \pi^\ast(x)$ for all $x >  \theta(\bar{a}_\infty,\bar{b}_\infty)$.

Combining Cases 1) and 2) we see that $\tilde{\pi}(x) \ge \pi^\ast(x)$ for all $x \in [0,\psi]$ as desired.%}
\remove{
\subsection{Proof of Lemma~\ref{lemma:apprx-inequality}}
\label{appendix:apprx-inequality}
Using Lemma~\ref{lemma:optimal-approx-0} and~\ref{lemma:pi-k-max} we can establish the following result
\[
\tilde{\pi}(x) \ge \pi^*(x)
\]
}
\vspace{-0.1in}
%\section{}
%\remove{
\section{Proof of Lemma~\ref{lemma:fixed-nash-properties-a-b}}
\label{appendix:fixed-nash-properties-a-b}
$(a)$ Notice that the mapping $a \mapsto \frac{1}{4-2pa}$ is monotonically increasing.
Further, $a_0 > a_{-1}$. Therefore the sequence $a_k, k\geq -1$ is monotonically increasing.
%It is also nonnegative, and so, lower bounded.
Hence it converges to $a_\infty$, the smallest fixed point
of $a = \frac{1}{4-2pa}$. \\
In the following we argue that $\tilde{a}_{\infty}<\frac{1}{3}$. By definition of $\tilde{a}_{\infty}$, it is enough to argue  \[6-2p<3\sqrt{4-2p}.\]
The above equation implies $p<\frac{3}{2}$, which is true always. Therefore, $a_{\infty}<\frac{1}{3}$.
\section{Proof of Lemma~\ref{lemma:fixed-game-no-caps}}
\label{appendix:fixed-game-no-caps}
 Using $\frac{1}{4}<\tilde{a}_{\infty}<\frac{1}{3}$, from Lemma~\ref{lemma:fixed-nash-properties-a-b} it can be seen that $\tilde{b}_{\infty}>0$. Therefore, $0<\tilde{a}_{\infty}x+\tilde{b}_{\infty}$.\\
 Let us now argue that $\tilde{a}_{\infty}x+\tilde{b}_{\infty}<\psi,~x\in[0,\psi]$. It is enough to show that $\tilde{a}_{\infty}\psi+\tilde{b}_{\infty}<\psi$. By definition of $\tilde{b}_{\infty}$ and $\tilde{a}_{\infty}$ it is equivalent to argue that the following holds
 \begin{align*}
     \tilde{b}_{\infty}&<\psi(1-\tilde{a}_{\infty})\\
     \tilde{a}_{\infty}\frac{2\psi(2-p)}{\sqrt{4-2p}}&<\psi(1-\tilde{a}_{\infty})\\
     \tilde{a}_{\infty}&<\frac{1}{{\sqrt{4-2p}}+1}
 \end{align*}
 It can be observed that
 \[
 \max{\sqrt{4-2p}+1}=3.
 \]
 Therefore, it enough to argue that $\tilde{a}_{\infty}<\frac{1}{3}$, which clearly holds true from Lemma~\ref{lemma:fixed-nash-properties-a-b}. Hence, the lemma holds.
 %}

\section{Proof of Lemma~\ref{lemma:pi-prime-x-less-than-xinfinity}}
\label{appendix:pi-prime-x-less-than-xinfinity}
Let us recollect the following result of Case~1 from  Appendix~\ref{appendix:fixed-nash-equilibrium}. 
\[
\pi'_0(x)=0, \forall x \le x_{\infty}
\]
In this subsection we do not have any constraint ($\tilde{a}_{\infty}x_{\infty}+\tilde{b}_{\infty}>x_{\infty}$) as in Case~2 of Appendix~\ref{appendix:fixed-nash-equilibrium}. Therefore, $\pi_1'(x),\forall x>x_{\infty}$ can be either $\pi_{\tilde{a}_{\infty},\tilde{b}_{\infty}}(x)$ or $\pi_{0,0}(x)$. Using~\eqref{eqn:pi-a-b(x)} we infer $\pi'_1(x),x>x_{\infty}$ will be either $\tilde{a}_{\infty}x+\tilde{b}_{\infty}$ or $\tilde{a}_{0}x+\tilde{b}_{0}$~(see~\eqref{eqn:fixed-a-case2-nash},~\eqref{eqn:fixed-b-case2-nash}). 

It should be realized that from the proof of Lemma~\ref{lemma:fixed-nash-properties-a-b} it is clear that $\tilde{a}_k<\tilde{a}_{\infty},\forall k \ge -1$. Using definition of $\tilde{b}_0,\tilde{b}_{\infty}$ it can be observed that $\tilde{b}_0<\tilde{b}_{\infty}$ is equivalent to 
\[
4+p>\frac{1}{\tilde{a}_{\infty}}
\]
As Lemma~\ref{lemma:fixed-nash-properties-a-b} states that $\frac{1}{4}<\tilde{a}_{\infty}<\frac{1}{3}$, the above inequality holds. This implies 
\[
\tilde{a}_{\infty}x+\tilde{b}_{\infty}>\tilde{a}_{0}x+\tilde{b}_{0},\forall x \in [0,\psi].
\]
Using~\eqref{eqn:h-a-b-definition}, we infer the following
\begin{equation}
\label{eqn:h00-hab-relation2}
h_{\tilde{a}_{\infty},\tilde{b}_{\infty}}(x,u) \le h_{\tilde{a}_{0},\tilde{b}_{0}}(x,u),\forall x,u \in [0,\psi]
\end{equation}
Using $\pi'_1(x)$ we can write $C_2(x)$ as follows \begin{align*}
C_2(x)=&\min\{(\psi+x)\psi,\min_{u \in [0,x_{\infty}]}h_{0,0}(x,u),\\
&\min_{u \in \mathbb{A}}h_{\tilde{a}_{0},\tilde{b}_{0}}(x,u),\min_{u \in \mathbb{B}}h_{\tilde{a}_{\infty},\tilde{b}_{\infty}}(x,u)\},
\end{align*}
where $\mathbb{A},\mathbb{B} \subset (x_{\infty},\psi]$.
Let us now determine $\pi'_2(x), \forall x \le x_{\infty}$. 
From Lemma~\ref{lemma:inequality-x-infinity}, we infer the following
\begin{align*}
    \psi(\psi+x) &\le \min_{u \in [0,\psi]} h_{\tilde{a}_{\infty},\tilde{b}_{\infty}}(x,u),\\
    & < \min\left\{\min_{u \in [0,x_{\infty}]}h_{0,0}(x,u),\min_{u \in \mathbb{B}}h_{\tilde{a}_{\infty},\tilde{b}_{\infty}}(x,u), \min_{u \in \mathbb{A}} h_{\tilde{a}_{0},\tilde{b}_{0}}(x,u) \right\},
\end{align*}
where the second inequality follows from~\eqref{eqn:h00-hab-relation} and~\eqref{eqn:h00-hab-relation2}.
Hence $\pi'_2(x)=0,\forall x \le x_{\infty}$. Similar argument can be followed to prove $\pi'_k(x)=0,\forall x \le x_{\infty}, k \ge 0$.

\section{Proof of Proposition~\ref{thm:fixed-nash-equilibrium}}
\label{appendix:fixed-nash-equilibrium}
%Let us recollect the following result of Case~1 from  Appendix~\ref{appendix:fixed-nash-equilibrium}. 
%\[
%\pi'_0(x)=0, \forall x \le x_{\infty}
%\]
Let us define
\begin{align}
\label{eqn:h-a-b-definition}
    &h_{a,b}(x,u)=(\psi+x-u)(\psi-u)+d+u(u+p(\psi-au-b)),
 \end{align}
  \begin{align*}
    %&\text{and }\\
\text{and }    &~~~~~\pi_{a,b}(x)=\argmin_{u \in [0,\psi]}h_{a,b}(x,u).
\end{align*}
The following can be verified
\begin{align}
\label{eqn:pi-a-b(x)}
    \pi_{a,b}(x)=\frac{x+(2-p)\psi+pb}{2(2-pa)}
\end{align}
\begin{align}
\label{eqn:expansion_hab}
h_{a,b}(x,u)=&\bigg(\frac{u^2}{2}-u\frac{(2-p)\psi+x+pb}{2(2-ap)}\nonumber\\&+\frac{d}{2(2-ap)}\bigg)2(2-ap)
+\psi(\psi+x).
\end{align}
From~\cite[Chapter~2, Proposition~1.2(b)]{b21}, $C_k(\cdot)$s converges to the optimal cost function $C(\cdot)$ and $\pi'_k(\cdot)$ converge to $\pi'(\cdot)$ irrespective of the initial function $C_0(x)$ in the value iteration. Now we analyze value iteration starting with a different function.
\remove{
\begin{align*}
  %  \label{eqn:nash-c-0-cost-fixed}
C_0(x)=&\min\{(\psi+x)\psi,\nonumber\\
&\min_{u \in [0,\psi]}\{(\psi-u)(\psi-u+x)+d\nonumber\\&+u(u+p(\psi
-\tilde{a}_{\infty}u-\tilde{b}_{\infty}))\}\}.
\end{align*}
It can be observed that the above equation can be written as follows}
\[
C_0(x)=\min\{(\psi+x)\psi,h_{\tilde{a}_{\infty},\tilde{b}_{\infty}}(x,\pi_{\tilde{a}_{\infty},\tilde{b}_{\infty}}(x))\}.
\]
Recall that $\pi'_0(x)$ is the solution to $C_0(x)$. To determine $\pi'_0(x)$, we need to find $ \argmin h_{\tilde{a}_{\infty},\tilde{b}_{\infty}}(x,\pi_{\tilde{a}_{\infty},\tilde{b}_{\infty}}(x))$. Realize that $\pi_{\tilde{a}_{\infty},\tilde{b}_{\infty}}(x)$ is $\argmin h_{\tilde{a}_{\infty},\tilde{b}_{\infty}}(x,\pi_{\tilde{a}_{\infty},\tilde{b}_{\infty}}(x))$. Using~\eqref{eqn:pi-a-b(x)}, it can be seen that
\begin{align*}
%\label{eqn:pi-a-b(x)}
    \pi_{\tilde{a}_{\infty},\tilde{b}_{\infty}}(x)=\frac{x+(2-p)\psi+p\tilde{b}_{\infty}}{2(2-p\tilde{a}_{\infty})}
\end{align*}
As the sequences $\tilde{a}_k,\tilde{b}_k,k \ge -1$ converge to $\tilde{a}_{\infty},\tilde{b}_{\infty}$ (see~\eqref{eqn:fixed-a-case2-nash},\eqref{eqn:fixed-b-case2-nash}). From Lemma~\ref{lemma:fixed-game-no-caps} we know that $\pi_{\tilde{a}_{\infty},\tilde{b}_{\infty}}(x) \in (0,\psi)$, therefore  we infer that
\begin{equation}
\label{eqn:pi-a-infinity-b-infinity}
\pi_{\tilde{a}_{\infty},\tilde{b}_{\infty}}(x)=\tilde{a}_{\infty}x+\tilde{b}_{\infty}.
\end{equation}
Now to determine $\pi'_0(x)$, we need the following lemma which is proved at the end of this proof.%~\cite[Appendix~II-C]{b29}.
\begin{lemma}
\label{lemma:inequality-x-infinity}
The following inequality holds if and only if $x \le x_{\infty}$.
\begin{equation}
\label{eqn:find-pi_0-prime}
 \psi(\psi+x) \le h_{\tilde{a}_{\infty},\tilde{b}_{\infty}}(x,\pi_{\tilde{a}_{\infty},\tilde{b}_{\infty}}(x)).
\end{equation}
\end{lemma}
Using Lemma~\ref{lemma:inequality-x-infinity} and Lemma~\ref{lemma:fixed-game-no-caps}, we infer
\begin{align}
\label{eqn:fixed-partial-game}
\pi'_0(x)=
\begin{cases}
0, &\text{ if } x\le x_{\infty}\\
\tilde{a}_{\infty}x+\tilde{b}_{\infty}, &\text{ otherwise}
\end{cases}
\end{align}
Using~\eqref{eqn:h-a-b-definition}, $\tilde{a}_{\infty}>0, \tilde{b}_{\infty}>0$ we infer the following
\begin{equation}
\label{eqn:h00-hab-relation}
h_{\tilde{a}_{\infty},\tilde{b}_{\infty}}(x,u) \le h_{0,0}(x,u),\forall x,u \in [0,\psi]
\end{equation}
Now from~\eqref{eqn:fixed-nash-j-prime-k} and~\eqref{eqn:fixed-partial-game}, the following can be written
%\vspace{-0.1in}
\begin{align*}
C_1(x)=\min\{(\psi+x)\psi,\min_{u \in [0,x_{\infty}]}&h_{0,0}(x,u),\\
&\min_{u \in [x_{\infty},\psi]}h_{\tilde{a}_{\infty},\tilde{b}_{\infty}}(x,u)\}.
\end{align*}
We would now determine $\pi'_1(x)$. Let us study the following two cases separately.
%\begin{enumerate}
%\item
\paragraph*{Case 1}~~$x \le x_{\infty}$

From Lemma~\ref{lemma:inequality-x-infinity}, we infer the following when $x \le x_{\infty}$
\begin{align*}
    \psi(\psi+x) &\le \min_{u \in [0,\psi]} h_{\tilde{a}_{\infty},\tilde{b}_{\infty}}(x,u),\\
    & < \min\left\{\min_{u \in [0,x_{\infty}]}h_{0,0}(x,u),\min_{u \in [x_{\infty},\psi]}h_{\tilde{a}_{\infty},\tilde{b}_{\infty}}(x,u)\right\},
\end{align*}
where the second inequality follows from~\eqref{eqn:h00-hab-relation}.
Hence $\pi'_1(x)=0,\forall x \le x_{\infty}$.
%\item
\paragraph*{Case 2}~~$x>x_{\infty}$

Note that $\frac{\tilde{b}_{\infty}}{1-\tilde{a}_{\infty}} \ge x_{\infty}$ implies $\tilde{a}_{\infty}x_{\infty}+\tilde{b}_{\infty}>x_{\infty}$. When $\tilde{a}_{\infty}x_{\infty}+\tilde{b}_{\infty}>x_{\infty}$ the following holds
\begin{align*}
    \min_{u \in [x_{\infty},\psi]}h_{\tilde{a}_{\infty},\tilde{b}_{\infty}}(x,u) &= \min_{u \in [0,\psi]} h_{\tilde{a}_{\infty},\tilde{b}_{\infty}}(x,u)\\
    & < \min\left\{(\psi+x)\psi,\min_{u \in [0,x_{\infty}]}h_{0,0}(x,u)\right\}
\end{align*}
Last inequality follows from Lemma~\ref{lemma:inequality-x-infinity} and~\eqref{eqn:h00-hab-relation}. Hence, $\pi'_1(x)=\tilde{a}_{\infty}x+\tilde{b}_{\infty}, \forall x > x_{\infty}$.
%\end{enumerate}

Combining both the cases $\pi'_1(x)=\pi'_0(x), \forall x \in [0,\psi]$. We can iteratively show that $\pi'_k(x)=\pi'_0(x), \forall x \in [0,\psi]$. Hence $\pi'(x)=\pi'_0(x)$.
%\remove{
\paragraph*{Proof of Lemma~\ref{lemma:inequality-x-infinity}}
%\label{proof:inequality-x-infinity}
Using~\eqref{eqn:expansion_hab}, we see that~\eqref{eqn:find-pi_0-prime} is equivalent to the following expression being greater than zero.
\[
\frac{\pi_{\tilde{a}_{\infty},\tilde{b}_{\infty}}(x)^2}{2}-\pi_{\tilde{a}_{\infty},\tilde{b}_{\infty}}(x)\frac{(2-p)\psi+x+p\tilde{b}_{\infty}}{2(2-\tilde{a}_{\infty}p)}+\frac{d}{2(2-\tilde{a}_{\infty}p)}
\]
As the sequences $\tilde{a}_k,\tilde{b}_k,k \ge -1$ converge to $\tilde{a}_{\infty},\tilde{b}_{\infty}$ (see~\eqref{eqn:fixed-a-case2-nash},\eqref{eqn:fixed-b-case2-nash}) the above inequality can be written as
\[
\frac{\pi_{\tilde{a}_{\infty},\tilde{b}_{\infty}}(x)^2}{2}-\pi_{\tilde{a}_{\infty},\tilde{b}_{\infty}}(x)(\tilde{a}_{\infty}x+\tilde{b}_{\infty})+d\tilde{a}_{\infty}>0
\]
Using~\eqref{eqn:pi-a-infinity-b-infinity}, the above equation can be written as
\begin{align*}
\frac{(\tilde{a}_{\infty}x+\tilde{b}_{\infty})^2}{2}-(\tilde{a}_{\infty}x+\tilde{b}_{\infty})^2+d\tilde{a}_{\infty}&>0,\\
\sqrt{2d\tilde{a}_{\infty}}&>\tilde{a}_{\infty}x+\tilde{b}_{\infty},\\
x&<\frac{\sqrt{2d\tilde{a}_{\infty}}-\tilde{b}_{\infty}}{\tilde{a}_{\infty}}.
\end{align*}
Recollect that $x_{\infty}=\frac{\sqrt{2d\tilde{a}_{\infty}}-\tilde{b}_{\infty}}{\tilde{a}_{\infty}}$. Therefore~\eqref{eqn:find-pi_0-prime} if and only if $x \le x_{\infty}$.%}

\color{black}
\end{document}